\documentclass[11pt,letterpaper]{article}
\usepackage[margin = 1in]{geometry}
\usepackage{graphicx}
\usepackage{amsmath,amsthm,amssymb}
\usepackage{thm-restate}
\usepackage[dvipsnames]{xcolor}
\usepackage[colorlinks=true, allcolors=blue]{hyperref}
\usepackage{cleveref}
\usepackage[ttscale=.85]{libertine}
\usepackage[T1]{fontenc}
\usepackage{enumitem}

\theoremstyle{plain}
\newtheorem{question}{Question}
\newtheorem{theorem}{Theorem}[section]
\newtheorem{lemma}[theorem]{Lemma}

\newtheorem*{claim*}{Claim}
\newenvironment{claimproof}[1][\proofname]{\proof[#1]}{\endproof}
\theoremstyle{definition}

\DeclareMathOperator{\dist}{dist}
\DeclareMathOperator{\diam}{diam}
\DeclareMathOperator{\lca}{lca}
\DeclareMathOperator{\conv}{conv}
\DeclareMathOperator{\desc}{desc}
\DeclareMathOperator{\cov}{cov}
\DeclareMathOperator{\cell}{cell}

\newcommand{\bigO}[1]{\mathcal{O}{\left(#1\right)}}
\newcommand{\Oh}{\mathcal{O}}
\newcommand{\bigOd}[1]{\mathcal{O}_d{\left(#1\right)}}
\newcommand{\bigOm}[1]{\Omega{\left(#1\right)}}

\newcommand{\bigTh}[1]{\Theta{\left(#1\right)}}
\newcommand{\eps}{\varepsilon}

\newcommand{\cC}{\mathcal{C}}
\newcommand{\cR}{\mathcal R}
\newcommand{\cQ}{\mathcal{Q}}
\newcommand{\cT}{\mathcal{T}}

\DeclareMathOperator\arsinh{arsinh}
\DeclareMathOperator\arcosh{arcosh}

\newcommand{\Hyp}{\mathbb{H}}
\newcommand{\Euc}{\mathbb{E}}
\newcommand{\Reals}{\mathbb{R}}
\newcommand{\Sph}{\mathbb{S}}
\newcommand{\distH}[1]{\dist_{\Hyp^{#1}}}
\newcommand{\distHd}{\distH{d}}

\renewcommand{\angle}{\sphericalangle}

\newcommand{\up}{\uparrow}
\newcommand{\down}{\downarrow}
\newcommand{\zu}{z^\up}
\newcommand{\zd}{z^\down}
\newcommand{\x}{x_{\min}}
\newcommand{\xm}{x_{\max}}

\newcommand{\stereo}{\pi_{\text{stereo}}}

\date{}
\title{Near-Optimal Dynamic Steiner Spanners\\ for Constant-Curvature Spaces}
\author{Sándor Kisfaludi-Bak\thanks{Department of Computer Science, Aalto University, Espoo, Finland, \textsf{sandor.kisfaludi-bak@aalto.fi}.
This work was supported by the Research Council of Finland, Grant 363444.}
\and
Geert van Wordragen\thanks{Department of Computer Science, Aalto
    University, Espoo, Finland, \textsf{geert.vanwordragen@aalto.fi}}
}

\begin{document}

\maketitle
\thispagestyle{empty}
\begin{abstract}\normalsize
    We consider Steiner spanners in Euclidean and non-Euclidean geometries. In the Euclidean setting, a recent line of work initiated by Le and Solomon [FOCS'19] and further improved by Chang et al.\ [SoCG'24] obtained Steiner $(1+\varepsilon)$-spanners of size $O_d(\varepsilon^{(1-d)/2}\log(1/\varepsilon)n)$, nearly matching the lower bounds of Bhore and T\'oth [SIDMA'22].
    
    We obtain Steiner $(1+\varepsilon)$-spanners of size $O_d(\varepsilon^{(1-d)/2}\log(1/\varepsilon)n)$ not only in $d$-dimensional Euclidean space, but also in $d$-dimensional spherical and hyperbolic space. For any fixed dimension $d$, the obtained edge count is optimal up to an $O(\log(1/\varepsilon))$ factor in each of these spaces. Unlike earlier constructions, our Steiner spanners are based on simple quadtrees, and they can be dynamically maintained, leading to efficient data structures for dynamic approximate nearest neighbours and bichromatic closest pair.

    In the hyperbolic setting, we also show that $2$-spanners in the hyperbolic plane must have $\Omega(n\log n)$ edges, and we obtain a $2$-spanner of size $O_d(n\log n)$ in $d$-dimensional hyperbolic space, matching our lower bound for any constant $d$. Finally, we give a Steiner spanner with \emph{additive} error $\varepsilon$ in hyperbolic space with $O_d(\varepsilon^{(1-d)/2}\log(\alpha(n)/\varepsilon)n)$ edges, where $\alpha(n)$ is the inverse Ackermann function.

    Our techniques generalize to closed orientable surfaces of constant curvature as well as to some quotient spaces.
\end{abstract}

\clearpage
\setcounter{page}{1}

\section{Introduction}
Given a set of points in some metric space, what is the most efficient way to represent all pairwise distances? One can use a complete geometric graph: the vertices are the given points, and the edges are weighted by the corresponding distance. A \emph{$t$-spanner} is some small subgraph $G$ of this complete graph such that the shortest path in $G$ among any pair of points is at most $t$ times longer than their distance. In Euclidean space, one can obtain $(1+\eps)$-spanners of linear size for any fixed constant $\eps>0$ \cite{Clarkson87,Keil88}. Spanners are not only interesting in their own right, but they are also used in a plethora of algorithmic applications, and they lie at the heart of many approximation algorithms~\cite{RaoS98,subsetspanner,DemaineHM10,BateniHM11,9719718}. There is an ongoing research effort to study spanners in more general settings, and to establish optimal trade-offs between their stretch ($t=1+\eps$) and their size (edge count).

Spanners and dynamic spanners have been studied extensively in Euclidean spaces~\cite{narasimhan2007geometric,Chew86,AryaDMSS95,DasN97} as well as in the more general setting of spaces with bounded doubling dimension~\cite{GaoGN06,GottliebR08a,ChanG09,Roditty12,ChanLNS15}. The \emph{doubling dimension} of a metric space is the minimum $k$ such that any ball in the metric space can be covered by $2^k$ balls of half its radius.
Today, we understand the optimal trade-offs between stretch and spanner size in Euclidean and doubling spaces.
In Euclidean space, several techniques for spanner constructions, such as greedy spanners~\cite{AlthoferDDJS93,narasimhan2007geometric} and locality-sensitive orderings~\cite{locsens,GaoH24} yield $(1+\eps)$-spanners of size (near) $\Oh_d(\eps^{1-d}n)$, which is optimal for any fixed dimension $d$~\cite{LeS19}.
Well-separated pair decompositions (WSPDs)~\cite{CallahanK95} give a slightly worse edge count of $\Oh_d(\eps^{-d} n)$, but have a simpler construction and analysis.

Many of these techniques have been ported to doubling spaces, and have been used successfully to create spanners (and approximation algorithms). Net trees~\cite{Har-PeledM06} play the role of a hierarchical decomposition of the space rather than quadtrees, and one can build WSPDs \cite{Talwar04,Har-PeledM06}, greedy spanners \cite{Smid09}, and locality-sensitive orderings \cite{locsens,FiltserL22,LaLe24}. Gottlieb and Roditty~\cite{GottliebR08} provide dynamic spanners in doubling spaces of dimension $d$ with $\eps^{-\bigO{d}}n$ edges and update time $\eps^{-\bigO{d}} \log n$. In this more general setting, the constants around the doubling dimension $d$ in the exponent cannot be nailed down precisely using current tools.

While the above spanners are already useful, one can obtain even smaller (sparser) spanners if we allow so-called Steiner points in the spanner graph, that is, new points from the ambient space. More precisely, a \emph{Steiner $(1+\eps)$-spanner} of a point set $P$ in a metric space $(X,\dist)$ is a subgraph $G$ of the complete geometric graph on $P\cup S$ (for some set $S\subset X\setminus P$) that has the spanner property on $P$: $\dist_G(u,v)\leq (1+\eps)\dist(u,v)$ for all $u,v\in P$. The points in $S$ are called \emph{Steiner points}. In $d$-dimensional Euclidean space (henceforth denoted by $\Euc^d$) it is possible to obtain Steiner $(1+\eps)$-spanners of size $\Oh_d(\eps^{(1-d)/2}\log(1/\eps) n)$ \cite{ChangC0MST24}. This is almost matched by a lower bound of~\cite{BhoreT22}, who prove that $\Omega_d(\eps^{(1-d)/2}n)$ edges are required. 

In this paper, we attempt to generalise Euclidean (Steiner) spanners to curved (non-Euclidean) spaces. In particular, we study the following question.

\begin{question}\label{q:tradeoff}
    What is the optimal trade-off between the stretch and the size of a (Steiner) $(1+\eps)$-spanner in non-Euclidean spaces?
\end{question}

To begin our investigation with the simplest non-Euclidean spaces, let us consider the two other geometries of constant curvature: spherical geometry,\footnote{Arguably, elliptic geometry is more appropriate, but less readily perceived. Our results also extend to elliptic geometry with a simple trick, which is presented in Section~\ref{sec:quotient}.} representing constant positive (sectional) curvature, and hyperbolic geometry, representing constant negative (sectional) curvature.
We note that spanner results apply regardless of the constant positive curvature of the spherical space, as a scaling of the metric can generate the same result for any other spherical space. The same observation holds in the hyperbolic setting. For these reasons, we work in $d$-dimensional spherical and hyperbolic space of sectional curvature $1$ and $-1$, and denote them by $\Sph^d$ and $\Hyp^d$, respectively.

Spherical space is naturally modelled by a $d$-dimensional sphere in $(d+1)$-dimensional Euclidean space, where distances are measured along the sphere.
Due to its positive curvature, spherical space grows more slowly than Euclidean space, especially at larger distances (e.g., the area of a disk of radius $r$ is $2\pi(1-\cos r)$ rather than $\pi r^2$).
In contrast, hyperbolic space grows more quickly than Euclidean space (e.g., the area of a disk of radius $r$ is now $2\pi(\cosh r - 1)$, which is $\Theta(e^r)$ for superconstant $r$).
It is commonly visualised using models that map it to Euclidean space with heavily distorted distances, similar to e.g.\ the stereographic and gnomonic projection used for spheres.
See for example \Cref{fig:binarytiling}, which shows a tiling of the hyperbolic plane with tiles that are, importantly, isometric.
Note the strong similarity to a binary tree.

\begin{figure}
    \centering
    \includegraphics[scale=.5]{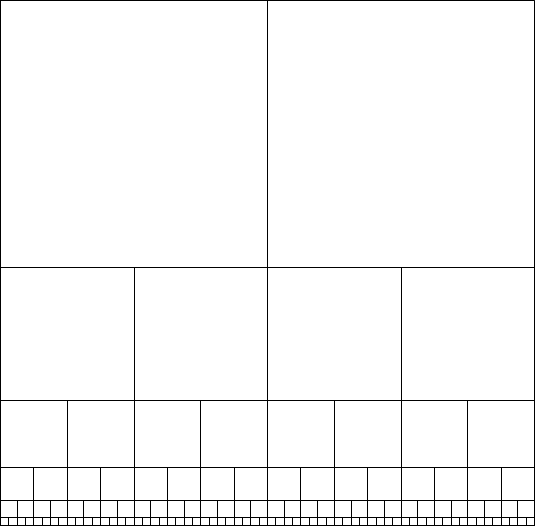}
    \caption{A patch of the binary tiling shown in the half-plane model. Note that all tiles are isometric.}
    \label{fig:binarytiling}
\end{figure}

One of the main motivations for studying curved spaces is that certain types of data can be much more accurately and efficiently modelled in a curved space than in (flat) Euclidean space.
In particular, constant-dimensional hyperbolic space already allows for exponential growth and unbounded doubling dimension, giving good embeddings for graphs with an inherent hierarchy, such as those based on the internet \cite{ShavittT04} and social networks \cite{VerbeekS14}.
This has attracted the attention of, for example, the machine learning \cite{NickelK18,GaneaBH18,NEURIPS2019_0415740e,PengVMSZ22}, graph visualisation \cite{lamping1995focus} and complex network modelling \cite{krioukov2010hyperbolic} communities.
From a more theoretical standpoint, negative curvature is one of the simplest ways a geometric space can be non-doubling.
Properly dealing with this requires new approaches that do not rely on doubling dimension and could inspire future work in other non-doubling spaces.
Additionally, every closed surface can be assigned constant curvature, making especially hyperbolic geometry a useful tool in computational topology \cite{VerdiereE10,ChangM22,VerdiereDD24}.

Both $\Sph^d$ and $\Hyp^d$ are locally Euclidean spaces, meaning that a small radius ball in each space can be embedded into $\Euc^d$ with very small distortion. On a larger scale however, they have a very different behaviour. In $\Sph^d$, we can still bound the doubling dimension by $\Oh(d)$, and thus doubling metric results apply, yielding spanners of size $\eps^{-\bigO{d}}n$. One would expect that the stronger Euclidean spanner and Steiner spanner constructions (of size $\Oh_d(\eps^{1-d} n)$ and $\Oh_d(\eps^{(1-d)/2}\log(1/\eps) n)$) should be possible in $\Sph^d$, however, proving it is non-trivial: the geometric transformations to convert between these cases introduce constant multiplicative error, thus a spanner property in $\Euc^d$ does not imply the spanner property in $\Sph^d$. Note that spanners of $\Sph^d$ with the same guarantees as the Euclidean constructions would be (nearly) optimal in $\Sph^d$, as the Euclidean lower-bound constructions of Bhore and Tóth~\cite{BhoreT22} can be ported to $\Sph^d$ inside a small neighbourhood. (Naturally, the Euclidean lower bounds hold in $\Hyp^d$ for the same reason.)

In $d$-dimensional hyperbolic space ($\Hyp^d$), the landscape is very different. First of all, $\Hyp^d$ is not a doubling space, so doubling space results do not apply. Second, as observed by~\cite{KrauthgamerL06,hyperquadtree}, there do not exist $t$-spanners for any $t<2$ of size less than $\binom{n}{2}$. In fact, many of the important techniques such as WSPDs and locality sensitive orderings are known to fail~\cite{hyperquadtree}. However, Krauthgamer and Lee~\cite{KrauthgamerL06} obtained a Steiner $(1+\eps)$-spanner with $\eps^{-\Oh_d(1)}n$ edges in a more general setting of visual geodesic Gromov-hyperbolic spaces. Their technique is based on the observation that over shorter distances, hyperbolic space has bounded doubling dimension and thus one can use doubling space techniques, while at larger distances hyperbolic space is sufficiently ``tree-like'' and one can obtain additive Steiner spanners. In an \emph{additive} (Steiner) spanner, the distance in the spanner is at most $\ell$ longer than the hyperbolic distance. Recently, Park and Vigneron~\cite{additivespanner} gave a Steiner spanner in $\Hyp^d$ with additive error $\Oh(k\log d)$ and $2^{\Oh(d)}n\lambda_k(n)$ edges, where $\lambda_k(n)$ is an extremely slow-growing $k$-th row inverse Ackermann function. Finally, Van Wordragen and Kisfaludi-Bak~\cite{hyperquadtree} gave a dynamic Steiner $(1+\eps)$-spanner in $\Hyp^d$ with $\Oh(\eps^{-d}\log(1/\eps)n)$ edges based on the construction of a new hyperbolic quadtree.
Their quadtree is shown in \Cref{fig:quadtrees}.
In short, it splits cells that are sufficiently small the same way as a Euclidean quadtree, while splitting larger cells into one top cell and a possibly unbounded number of bottom cells.
From their Steiner spanner, \cite{hyperquadtree} also obtain linear size $(2+\eps)$-spanners, leaving the matter of $2$-spanners open.

\begin{figure}
    \centering
    \includegraphics{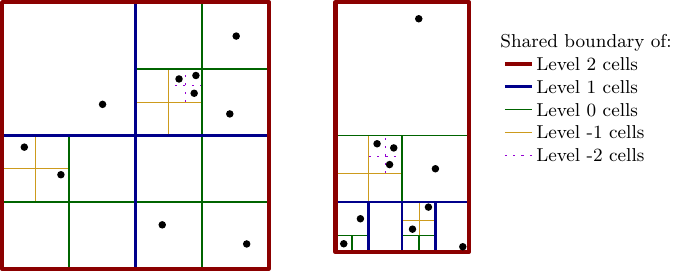}
    \caption{A Euclidean quadtree and the hyperbolic quadtree of \cite{hyperquadtree}.}
    \label{fig:quadtrees}
\end{figure}

\begin{question}\label{q:twospan}
    Are there $2$-spanners of size $\Oh_d(n)$ in $\Hyp^d$?
\end{question}

Hyperbolic spaces can exhibit very different behaviours than Euclidean spaces due to their non-doubling nature at larger scales. The tree-likeness of hyperbolic spaces has been used to obtain stronger structural and algorithmic results for several geometric problems~\cite{hyperSepSODA,hyperTSP,structureindependence} in $\Hyp^d$ than what is possible in~$\Euc^d$. It is thus natural to ask the following question.

\begin{question}\label{q:hypersteiner}
    Are there (Steiner) spanners in $\Hyp^d$ that have even stronger guarantees than their Euclidean counterparts for points that are pairwise distant?
\end{question}

The existence of additive Steiner spanners already answers this question affirmatively, but finding the best trade-off between the additive error of a Steiner spanner and its edge count remains open.

Outside the specific hyperbolic setting, there is growing interest in obtaining spanners in other non-doubling spaces, in particular, metric spaces induced by planar, minor-free, or bounded treewidth graphs as well as polyhedral terrains or surfaces. In recent work, the technique of \emph{tree covers} has gained momentum.
A tree cover can be thought of as a set of weighted trees whose union gives a (Steiner) spanner.
More precisely, set $\cT$ is a tree cover for a point set $P$ in metric space $X$ if for any pair of points in $P$ there is a tree in $\cT$ where their distance is at most $1+\eps$ times larger than their distance in $X$, while their distance in all trees of $\cT$ is at least that in $X$.
In geometric settings, the trees are typically required to be geometric graphs, i.e.\ any Steiner points are points in $X$ and the edge weights match the distances in $X$.
Chang~et~al.~\cite{ChangCLMST23} gave (Steiner) tree covers in planar and treewidth $t$ metrics, which yield Steiner $(1+\eps)$-spanners of size $\bigO{\eps^{-3} \log\frac{1}{\eps}n}$ for planar metrics and $2^{(t/\eps)^{\bigO t}}n$ for treewidth-$t$ metrics. Even more recently, Bhore~et~al.~\cite{BhoreKK0LPT25} obtained a Steiner $(1+\eps)$-spanner on planar metrics and on $2$-dimensional polyhedral surfaces of bounded genus with $\bigO{\frac{1}{\eps} \log\frac{1}{\eps} \cdot \log\frac{\alpha(n)}{\eps} \cdot n}$ edges, where $\alpha(n)$ is the inverse Ackermann function. These settings vastly generalize the finite metrics induced by point sets in $\Hyp^2$ and thus the trade-offs in these settings are not directly comparable.
However, the hyperbolic plane does have many of the important properties of more general planar metrics: for example, there is a natural lower bound of $2$ for the stretch of (non-Steiner) spanners of subquadratic size both in planar metrics and in $\Hyp^2$~\cite{BhoreKK0LPT25,hyperquadtree}.

\paragraph*{Our contribution}

Our main contribution is the following theorem.

\begin{theorem}\label{thm:main}
    Let $\eps \in (0,\frac12]$ and $P$ be a set of $n$ points in $d$-dimensional Euclidean, hyperbolic or spherical space.
    There is a Steiner $(1+\eps)$-spanner for~$P$ with $\bigOd{\eps^{(1-d)/2} \log \frac{1}{\eps} \cdot n}$ edges.
    We can maintain it dynamically, such that each point insertion or deletion takes $\bigOd{\eps^{(1-d)/2} \log \frac{1}{\eps} \cdot \log n}$ time.
\end{theorem}

Our edge count nearly matches the lower bound $\bigOm{\eps^{(1-d)/2}n}$ given by Bhore and Tóth~\cite{BhoreT22} for Euclidean Steiner spanners, up to the factor $\log\frac{1}{\eps}$. Our $\Oh_d$ notation hides a factor of $d^{\bigO{d}}$. Note that the Euclidean lower bound is also a valid lower bound in $\Sph^d$ and $\Hyp^d$, as both spaces are locally Euclidean, meaning that any hard instance in $\Euc^d$ can be embedded both in $\Sph^d$ and $\Hyp^d$ with arbitrarily small distortion.

Thus, we nearly settle \Cref{q:tradeoff} in the three geometries cases of constant dimension, up to the $\log 1/\eps$ term.  In Euclidean space our edge count matches a recent result of Chang~et~al.~\cite{ChangC0MST24}, however, our technique is completely different and arguably simpler. The main advantage is that our spanner can be maintained dynamically when points are added to $P$ or deleted from $P$, using $\Oh_{d,\eps}(\log n)$ time updates. As a result of the dynamicity and the simple structure, our spanner can readily act as a data structure for classic query problems such as dynamic approximate nearest neighbours and (dynamic approximate) bichromatic closest pair.
From the Steiner spanner constructions, one can also directly get (dynamic) Steiner tree covers of size $\bigOd{\eps^{(1-d)/2} \log \frac{1}{\eps}}$, which generalise the result of \cite{ChangC0MST24}.

The main challenge in establishing \Cref{thm:main} lies in the hyperbolic variant. While our spanner for $\Sph^d$ can be relatively easily derived from our Euclidean construction using geometric transformations, our spanner for $\Hyp^d$ is substantially different. In our approach, we use Euclidean (doubling space) ideas for small distances and tree-based ideas (constant-additive spanner) for sufficiently large distances, but in the hyperbolic setting, this leaves a gap for point pairs at distances between $1$ and $\frac{1}{\eps}$. This is the most challenging regime and requires mixing ideas from both approaches, as well as new ideas.  We heavily build on the techniques of~\cite{hyperquadtree}, and end up improving the size of the hyperbolic Steiner spanner from $\Oh_d(\eps^{-d}\log(1/\eps) n)$ (given in \cite{hyperquadtree}) to $\Oh_d(\eps^{(1-d)/2}\log(1/\eps) n)$.

Note that the spanners of \Cref{thm:main} are bipartite in case of $\Euc^d$ and $\Sph^d$, and we are also able to design bipartite Steiner spanners in $\Hyp^d$ for any $d\geq 3$ with the same properties as stated in \Cref{thm:main}. Surprisingly, we prove that a bipartite Steiner spanner of size $o(n/\eps)$ does not exist in $\Hyp^2$.
As a corollary of this lower bound, we obtain an $\Omega(n\log n)$ lower bound for the edge count of a hyperbolic (non-Steiner) $2$-spanner in the plane.
We prove that there exist $2$-spanners in $\Hyp^d$ of size $\Oh_{d,\eps}(n\log n)$, matching our lower bound, and settling \Cref{q:twospan}.

\begin{restatable}{theorem}{twospanner}
\label{thm:twospanner}
    Let $P \subset \Hyp^d$ be a set of $n$ points.
    We can construct a $2$-spanner for $P$ with $d^{\bigO{d}} \cdot n \log n$ edges in $d^{\bigO{d}} \cdot n \log n$ time.
\end{restatable}

By incorporating so-called \emph{transitive closure spanners} on certain trees obtained from hyperbolic quadtrees, we are also able to get Steiner spanners with $\eps$ \emph{additive} error for any $\eps>0$ whose size is very close to linear in $n$: it has size $\Oh_{\eps,d}(n\log \alpha(n))$, where $\alpha$ is the inverse Ackermann function.

\begin{restatable}{theorem}{additivespanner}\label{thm:additive}
    Let $\eps \in (0, \frac{1}{2}]$ and $P$ be a set of $n$ points in $\Hyp^d$.
    An $\eps$-additive Steiner spanner for~$P$ with $d^{\bigO{d}} \eps^{(1-d)/2} \log\frac{\alpha(n)}{\eps} \cdot n$ edges can be constructed in $d^{\bigO{d}} \eps^{(1-d)/2} \log\frac{\alpha(n)}{\eps} \cdot n \log n$ time.
\end{restatable}

While no direct lower bounds are known for additive hyperbolic spanners, this stronger variant of an additive spanner is in stark contrast with the Euclidean setting where no subquadratic-size additive (Steiner) spanners can exist. In terms of \Cref{q:hypersteiner}, this is further evidence that pairwise distant points in $\Hyp^d$ are structurally simpler than their Euclidean counterparts. It is an interesting open question whether one can improve the edge count by removing the $\log \alpha(n)$ term from the bound of \Cref{thm:additive}.

Finally, we apply our spanners to two different problems. First, we observe that our spanner construction easily extends to certain closed Riemannian manifolds of constant curvature: in particular, closed constant-curvature orientable surfaces and higher-dimensional manifolds obtained as a quotient of the three base geometries using a finite isometry group.

Next, we demonstrate the usefulness of our spanners by establishing \emph{dynamic} data structures for approximate nearest neighbours (ANN) and bichromatic closest pair.
In the $(1+\eps)$-approximate bichromatic closest pair problem, we are given two point sets $R$ and $B$ and wish to find a pair $(r,b) \in R \times B$ whose distance is at most $1+\eps$ more than that of the closest pair between $R$ and $B$.
The dynamic version requires that the approximate bichromatic closest pair can be recalculated efficiently after point insertions and removals in $R$ and $B$.
Here, we improve over the result of \cite{GaoH24} in Euclidean space, which uses $\bigOd{\eps^{1-d} \log(1/\eps) n}$ space and handles updates in $\bigOd{\eps^{1-d} \log(1/\eps) \log n}$ time, with a data structure that uses $\bigOd{\eps^{(1-d)/2} \log(1/\eps) n}$ space and handles updates in $\bigOd{\eps^{(1-d)/2} \log(1/\eps) \log n}$ time.
We also improve on the result of \cite{hyperquadtree} in hyperbolic space and to the best of our knowledge give the first data structure for approximate bichromatic closest pair that is explicitly for spherical space.
The same holds for (dynamic) ANN in hyperbolic and spherical space.

\subsection{Technical overview and key ideas}

All of our results are built on quadtrees. A quadtree is a simple hierarchical decomposition that is often used for geometric approximation in Euclidean space. More precisely, each of our spanners are built on $\Oh(d)$ different shifted quadtrees. It is a well-known fact that in Euclidean space, one can find $\Oh(d)$ shifts with the following key property.
\begin{quote}
    (*) Any pair of points $p,q$ is contained in a cell of diameter $|pq|\cdot \Oh(d\sqrt{d})$.
\end{quote}

To get a set of shifted quadtrees in spherical space, we take two stereographic projections and construct shifted Euclidean quadtrees in both of these projections.
For a pair of points that are not too distant, one of the stereographic projections will always give constant distortion compared to Euclidean space, and Property (*) can be achieved in one of the shifts for this projection.
In hyperbolic space we use the quadtree of \cite{hyperquadtree} that already comes with the required guarantee.
For small enough distances, hyperbolic space also has a constant distortion compared to Euclidean space and as a result, the quadtrees share most of the properties of the Euclidean quadtree in constant-size neighbourhoods.

While constant distortion is enough to get constant-approximations, it is not enough for a $(1+\eps)$-approximation.
For this, we prove in \Cref{lem:split} that a simple approximation result in Euclidean space also carries through to spherical and hyperbolic space. The following lemma is the main reason why (at least in constant-diameter regions) we manage to get a dependency on $\eps^{(1-d)/2}$ instead of $\eps^{1-d}$.

\begin{restatable*}{lemma}{split}\label{lem:split}
    Let $p,q,x \in \mathbb{X}^d$ for $\mathbb{X}\in \{\Hyp,\Euc,\Sph\}$ and $y$ the closest point to $x$ on segment $pq$.
    Assume $\eps \in (0,1]$.
    If $|xy| \leq \sqrt \eps \min\{|py|,|qy|\}$, then (i) $|px| + |xq| \leq (1 + \eps) |pq|$ if\, $\mathbb{X} \in \{\Euc,\Sph\}$ and (ii) $|px| + |xq| \leq (1 + e^\Delta \eps) |pq|$ if\, $\mathbb{X} = \Hyp$ and $|pq| \leq \Delta$.
\end{restatable*}

With this lemma at hand, we can now give a Steiner $(1+\eps)$-spanner.
For any two points $p$ and $q$, we consider the shifted quadtree where the cell satisfying property (*) is found.
We ensure that within this cell, both $p$ and $q$ have some point nearby, say $p'$ and $q'$, respectively, such that $p'$ and $q'$ are both connected to a Steiner point that is chosen from a grid-like point set on a cell boundary separating $p'$ and $q'$.
This Steiner spanner proves \Cref{thm:main} for $\Sph^d$, $\Euc^d$ and small-diameter point sets in~$\Hyp^d$.

\begin{figure}
    \centering
    \includegraphics{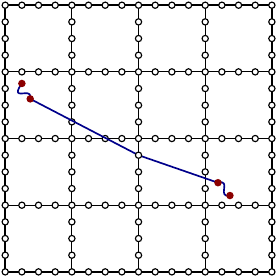}
    \qquad\qquad
    \includegraphics{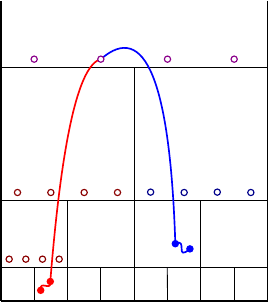}
    \caption{The Steiner spanner in $\Euc^2$ and the bipartite Steiner spanner in $\Hyp^2$.
    }
    \label{fig:steinerspanners}
\end{figure}

By construction this is a \emph{bipartite} Steiner spanner (i.e.\ every edge is between an input point and a Steiner point).
This makes it easy to turn it into a spanner without Steiner points.
For that, we take each Steiner point $s$ and connect all input points that were connected to $s$ to each other instead.
This maintains the same bound on the stretch, but adds more edges based on the degrees of the Steiner points.
In this case, the Steiner points have small degree, which leads to a total of $\bigOd{n/\eps^d}$ edges in \Cref{thm:doublingspanner}.

\subparagraph{Hyperbolic bipartite Steiner spanner.}
For super-constant hyperbolic distances, first of all the hyperbolic quadtree gets meaningfully different from a Euclidean quadtree.
The tiling given by level $0$ of the hyperbolic quadtree in $\Hyp^2$ resembles a binary tree. Cells of a given level $\ell$ correspond to subtrees of height $2^\ell$, therefore splitting a quadtree cell into its children corresponds to cutting the tree into subtrees at half its depth.
On top of this, many of the results that made the first construction work break down as the exponential expansion of hyperbolic space becomes an issue.
In particular, \Cref{lem:split} does not work for super-constant hyperbolic distances.
This is not a flaw with the proof: as the distances grow the hyperbolic Pythagorean theorem converges to the triangle inequality, which precludes us from proving something stronger for large hyperbolic distances.
At medium hyperbolic distances (between $1$ and $1/\eps$), we can give a result similar to \Cref{lem:split} which enables the dependence $\eps^{(1-d)/2}$ instead of $\eps^{1-d}$ when $d\geq 3$. 
\begin{restatable*}{lemma}{Hsplit}\label{lem:Hsplit2}
    Let $p,q \in \Hyp^d$ be separated from each other by hyperplane $H$ that has distance at least $2$ from both.
    For a point $x \in H$ at distance $\delta$ from $pq$, we get $|px| + |xq| \leq |pq| + \delta^2$.
\end{restatable*}
Note that when $|pq| = \bigTh{1}$ and $\delta = \sqrt\eps |pq|$ this is very similar result to \Cref{lem:split}, while when $|pq| = \bigTh{1/\eps}$ and $\delta = \eps |pq|$ we get the same guarantee from the triangle inequality; the most interesting regime of the lemma is in between these two values.
The hyperbolic bipartite Steiner spanner now follows from a similar construction as the one for smaller distances, with two major changes.
First, we replace \Cref{lem:split} with \Cref{lem:Hsplit2}.
Second, we only connect each point to its ``ancestors'' in the tiling (here one should think of ancestors in the binary tree), as opposed to all tiles, because in the hyperbolic quadtree there can be unboundedly many sibling tiles inside a parent cell.
With this Steiner spanner we can prove \Cref{thm:main} for all cases except for super-constant distances in $\Hyp^2$, where we get $\bigO{n \log(1/\eps)/\eps}$ instead of $\bigO{n \log(1/\eps) / \sqrt{\eps}}$ edges.
The coming lower bound shows that, surprisingly, having a weaker result for specifically $\Hyp^2$ is unavoidable unless we go for a very different, non-bipartite construction.

First, \Cref{thm:2spanner} gives a simple alternative procedure to turn a bipartite Steiner spanner into a spanner without Steiner points (as already observed by \cite{hyperquadtree}).
The procedure maintains the same edge bound, but in the worst case it doubles the stretch from $(1+\eps)$ to $2+2\eps$. As a result, we improve the $(2+\eps)$-spanner of~\cite{hyperquadtree}, but perhaps more importantly it lets us apply the coming spanner lower bound to bipartite Steiner spanners. Moreover, we will use the same strategy to make a (non-Steiner) $2$-spanner.

\subparagraph{Lower bound and 2-spanner.}
For the lower bound, we show that a $(2+\frac{1}{\ln n})$-spanner for a set of~$n$ points on a circle of radius $3 \ln n$ already requires $\bigOm{n \log n}$ edges.
This implies three things for these point sets:
\begin{itemize}
\item a $(2+\eps)$-spanner needs $\bigOm{n / \eps}$ edges, \item a bipartite Steiner $(1+\eps)$-spanner needs $\bigOm{n / \eps}$ edges (through \Cref{thm:2spanner}), and 
\item a $2$-spanner needs $\bigOm{n \log n}$ edges.
\end{itemize}
The argument goes roughly as follows.
First of all, the smallest and largest distance among the points differ only by a multiplicative factor close to $\frac{3}{2}$, which lets us conclude that the path between any two points $p$ and $q$ in the $(2+\frac{1}{\ln n})$-spanner can have at most $2$ hops. We denote by $m$ the middle vertex of the $2$-hop path.
Next, we observe that the location of $m$ is highly constrained:
the number of points on the circle separating $m$ from $p$ and $q$ must be at most a constant times larger than the number of points separating $p$ from $q$.
This observation lets us partition all pairs of points into $\bigTh{\log n}$ distance classes, each of which requires $\bigTh{n}$ edges, giving the $\bigOm{n \log n}$ lower bound.

We match this lower bound with a 2-spanner that has $\bigOd{n \log n}$ edges (in hyperbolic space of any dimension).
To get the 2-spanner, we use both ways of getting a spanner from a Steiner spanner.
For small distances we can get a $(1+\eps)$-spanner from \Cref{thm:doublingspanner}.
For larger distances, we first construct a new bipartite Steiner spanner that gives constant additive error and has $\bigOd{n \log n}$ edges.
We then turn this into a $2$-spanner, without the additional additive error: we bound the stretch more precisely using the specific geometric configuration.

\subparagraph{Non-bipartite Steiner spanners.}
To bypass the lower bound, we then work towards a non-bipartite Steiner spanner (i.e., one where Steiner points can be connected to Steiner points). This is our most involved construction, but arguably it is still relatively easy to describe.
The key to making this spanner work is to make use of the exponential divergence of lines in hyperbolic space: two long segments whose endpoints lie close together get exponentially close near their middle.
We capture this property in \Cref{lem:lambert}.
Directly from \Cref{lem:lambert}, we get \Cref{lem:verylargedist} which gives a quasi-isometric embedding of any hyperbolic point set into a \emph{union} of $\bigO{d}$ trees, similar to the much more involved quasi-isometric embedding of any hyperbolic metric space into a \emph{product} of $d$ trees \cite{EmbTreeProduct,EmbTreeProduct2}.
For our purposes, it implies a Steiner $(1+\bigO{\eps})$-spanner with only $\bigO{dn}$ edges for hyperbolic point sets with minimum distance $\frac{1}{\eps}\log\frac{1}{\eps}$.
To obtain a Steiner spanner for smaller distances, we use a more elaborate placement of Steiner points combined with ideas from the bipartite Steiner spanner and the bound of \Cref{lem:Hsplit2}.
This finally proves that our main theorem \Cref{thm:main} holds in all cases.

We can also get a Steiner spanner with \emph{additive} error $\eps$ by using the same general idea as in \Cref{thm:Hmain}. One of the issues here is that we accumulate some error in each hop along a spanner path, and we do not have a bound on the number of hops a spanner path might need to take. This problem can be mitigated by using a so-called transitive closure spanner to reduce the number of hops, and get the same error no matter how (hop-)distant two points are.

\subparagraph{Dynamic data structures.}
We can maintain our three multiplicative Steiner spanners dynamically by using the so-called Z-order or L-order for each of our shifted quadtrees. These are specific orderings of the input point set based on post-order traversals of the respective quadtree, i.e., our data structure is essentially $\bigO{d}$ differently ordered lists of our $n$ input points, as seen in locality-sensitive orderings~\cite{locsens,GaoH24}.

Comparing two points based on such an order is equivalent to asking which comes first if we add both to a quadtree and do a depth-first traversal of it, but these comparisons can be done in $\bigO{d}$ time without explicitly computing the quadtree.
The multiplicative Steiner spanner constructions are such that for any input point, the edges connected to it purely depend on its neighbours in the Z- or L-order, which allows for efficient updates.

\subparagraph{Organisation.} We explain the quadtrees and their main properties in \Cref{sec:prelims}, and introduce some key notation. \Cref{sec:doublingspace} has our first spanner construction for the Euclidean, spherical, and constant-diameter hyperbolic space settings. In \Cref{sec:hyperbolic} we tackle the large-scale hyperbolic setting, and this section contains most of our technical contributions. In \Cref{sec:applications} we list some direct consequences of our spanners, including spanners in some quotient spaces and closed surfaces, as well as data structures for dynamic approximate nearest neighbours and bichromatic closest pair. Finally, \Cref{sec:conclusion} concludes the paper by proposing some intriguing open problems.

\section{Preliminaries}\label{sec:prelims}
Throughout the paper, we will use $pq$ to denote the shortest curve (\emph{segment}) connecting $p$ and $q$ in the current geometry.
In particular, this means it may not appear as a straight line segment in the used model of that geometry.
Additionally, we use $|pq|$ for the length of this segment, again measured in the current geometry.
However, since $\Sph^d$ and $\Hyp^d$ are not normed spaces, the norm $\|\cdot\|$ of a vector in $\Reals^d$ is always unambiguously its Euclidean norm.
We will frequently need both the base-$2$ and the natural logarithm.
To that end, we reserve $\log$ for the former and use $\ln$ for the latter.

Our algorithmic results rely on compressed quadtree operations and therefore require access to floor and bitwise operations.
Apart from that, they hold both in the real RAM and word RAM model.
In particular, we do not need (hyperbolic) trigonometric functions despite these showing up in the distance formulas for spherical and hyperbolic space, since we only ever compare two distances between pairs of input points to each other. (In case of word RAM, we need to assume that points of $\Hyp^d$ are given with coordinates in the half-space model, as we do not have the capacity to convert between different models of hyperbolic geometry.)
In $\Sph^d$, any reasonable representation\footnote{Note that unlike $\Hyp^d$, the conversions between different representations of $\Sph^d$ and distance computations can typically be carried out with exponentially small error, which is sufficient for constructing $(1+\eps)$-spanners.} of points can work with our algorithm. We will have the points represented in polar coordinates.

\paragraph{Spherical geometry.}
Spherical geometry $\Sph^d$ can be modelled by the unit sphere in $\Euc^{d+1}$, where distances are measured along the sphere.
The segment (shortest curve) connecting two points is always an arc of the great circle through those points.

A common way to represent $\Sph^d$ is by a \emph{stereographic projection}, which maps all except one point of $\Sph^d$ to $\Euc^d$.
For this, first take $\Sph^d$ as a unit sphere in $\Euc^{d+1}$, then fix one point $o \in \Sph^d$ and let $H$ be the Euclidean hyperplane tangent to the point opposite $o$.
For any point $q \in \Sph^d \setminus \{o\}$, its stereographic projection is now given by the intersection of $oq$ with $H$.

\paragraph{Hyperbolic geometry.}
In this paper, we use the half-space model of hyperbolic geometry.
This assigns any point $p \in \Hyp^d$ coordinates $(x(p), z(p)) \in \Reals^{d-1} \times \Reals_{> 0}$ in the upper half-space, which also lets us refer to the positive $z$ direction as ``up''.
Important to note is that hyperbolic distances are not given by the Euclidean distances in the model, which also means that the shortest curve $pq$ between points $p,q \in \Hyp^d$ is often not a Euclidean segment.
In general, $pq$ is modelled by an arc from a Euclidean circle perpendicular to the Euclidean hyperplane $z=0$ (in particular its centre lies at $z=0$) and $pq$ has hyperbolic length
\[ |pq| = 2\arsinh\sqrt{\frac{\|x(p) - x(q)\|^2 + (z(p) - z(q))^2}{4z(p)z(q)}}, \]
but if $x(p) = x(q)$ then $pq$ is a vertical Euclidean line segment and $|pq| = \left|\ln\frac{z(p)}{z(q)}\right|$.
Note here that applying any isometry of $\Euc^{d-1}$ to $x(p)$ and $x(q)$ yields the same distance, as does multiplying all of $x(p), x(q), z(p)$ and $z(q)$ by the same positive factor.
Vertical Euclidean hyperplanes are also hyperbolic hyperplanes (as with vertical segments), but horizontal ones are not.
These are \emph{horospheres}, but we will not directly use any property of horospheres other than them being distinct from hyperplanes.

\paragraph{Hyperbolic quadtree.}
Here we describe the quadtree of \cite{hyperquadtree} while introducing new notation that will be useful throughout this paper.
Instead of being based on hypercubes, its cells are \emph{cube-based horoboxes}.
In the half-space model, a cube-based horobox $C$ is a Euclidean axis-parallel box with corners $(\x(C),\, \zd(C))$ and $(\xm(C),\, \zu(C))$ with its \emph{width} given by $\frac{\xm(C) - \x(C)}{\zd(C)} = (w(C),\dots,w(C))$ and \emph{height} $h(C) = \log\frac{\zu(C)}{\zd(C)}$.
This definition ensures that any two cube-based horoboxes with the same width and height are isometric.
Just like any box can tile Euclidean space, any horobox can tile hyperbolic space.

Instead of being based on the grid (the tiling of $\Euc^d$ with isometric hypercubes) like the Euclidean quadtree, the hyperbolic quadtree is based on the \emph{binary tiling}.
To construct a binary tiling, we start with the cube-based horobox $C$ with $\x(C) = (0, \dots, 0)$, $\zd(C) = 1$, $w(C) = \frac{1}{\sqrt{d-1}}$ and $h(C) = 1$.
Now we apply the isometry $T_{\sigma, \tau}(x,z) = \sigma \cdot (x + \tau, z)$ for all values $\tau = a / \sqrt{d-1}$ with $a \in \mathbb Z^{d-1}$ and $\sigma = 2^b$ with $b \in \mathbb Z$.
This gives a tiling that covers all of $\Hyp^d$ with isometric tiles; see \Cref{fig:binarytiling}.
Note that the two-dimensional tiling resembles a binary tree, hence the name.
In general the $d$-dimensional tiling bears resemblance to a $(d-1)$-dimensional Euclidean quadtree, where each ``layer'' of the tiling at a fixed $z$ value is a level of this quadtree.

We formalise this tree structure by thinking of the binary tiling as the infinite arborescence $T^0$, which has the tiles as its vertices and a directed edge from any tile to the tile directly above it.
We can then also generalise this to $T^\ell$ for $\ell \in \mathbb N$.
The tiling for $T^\ell$ is constructed by taking the cube-based horobox $C$ with again $\x(C) = (0, \dots, 0)$ and $\zd(C) = 1$, but now $w(C) = 2^{2^\ell-1} / \sqrt{d-1}$ and $h(C) = 2^\ell$.
Then, we apply $T_{\sigma, \tau}$ for all values $\tau = a \cdot w(C)$ with $a \in \mathbb Z^{d-1}$ and $\sigma = 2^{b \cdot h(C)}$ with $b \in \mathbb Z$.
The arborescence $T^\ell$ follows as before, where it is worth noting that the in-degree of vertices of $T^\ell$ has grown from $2^{d-1}$ to $2^{(d-1) \cdot 2^\ell}$.

The tiles are chosen exactly so that each tile of $T^{\ell+1}$ is the union of a tile of $T^\ell$ with its children.
Thus, the tilings for $T^\ell$ can be used to define the hyperbolic quadtree for levels $\ell \geq 0$.
Levels $\ell < 0$ are now constructed by splitting each cell at level $\ell+1$ up with a Euclidean quadtree split.
See \Cref{fig:quadtrees} for an illustration.
We will use the following properties of the hyperbolic quadtree from \cite{hyperquadtree} (with the source in \cite{hyperquadtree} of the property in parentheses):
\begin{lemma}[\cite{hyperquadtree}]\label{lem:hyperquadtree}
There exist quadtrees in $\Hyp^d$ with the following properties.
    \begin{enumerate}[label=(\roman*)]
        \item\label{lem:Hshift}
        For any $\Delta \in \Reals_+$, there is a set of at most $3d+3$ infinite hyperbolic quadtrees such that any two points $p,q \in \mathbb H^d$ with $\distHd(p,q) \leq \Delta$ are contained in a cell with diameter $\bigO{d\sqrt{d}} \cdot \distHd(p,q)$ in one of the quadtrees.
        (Theorem~1)
        
        \item\label{lem:Hcelldiam}
        If $C$ is a cell at level $\ell$, then $\diam(C) = \bigTh{2^\ell}$.
        (Theorem~7)

        \item\label{lem:Hdegree}
        A quadtree cell $C$ has $\max(2^d,2^{\bigO{ d \cdot \diam(C) }})$ children.
        (Theorem~7)

        \item\label{lem:Lorder}
        For two points $p,p' \in \Hyp^d$, we can check which comes first in the L-order for $\cQ^d_\infty$ by using $\bigO{d}$ floor, logarithm, bitwise logical and standard arithmetic operations.
        (Lemma~12)

        \item\label{lem:Hstarshaped}
        For any $p,q \in \Hyp^d$ and any hyperbolic quadtree cell $C$ that contains both, if $p$ lies below $z = \zd(C) \cdot 3^{h(C)/2}$ then the geodesic between $p$ and $q$ is completely contained in~$C$.
        (Lemma~18)
    \end{enumerate}
\end{lemma}

Since we will make frequent use of trees and arborescences, we also introduce some related notation.
Given two cells $C, C' \in V(T)$, we let $\lca(C, C', T)$ denote their lowest common ancestor in $T$.
Given a point $p \in \Hyp^d$, we let $\cell(p,T)$ denote the cell of $T$ containing $p$, and given a second point $q \in \Hyp^d$ we use the shorthand $\lca(p,q,T) = \lca(\cell(p, T), \cell(q, T), T)$.

Additionally, the structure of the quadtree and the tilings $T^\ell$ make it useful to occasionally look at the projections $\pi_x$ and $\pi_z$, which map points to only their $x$- or $z$-coordinates.
We also use the modified mappings $\tilde\pi_x(p) = \sqrt{d-1} \cdot \pi_x(p)$ and $\tilde\pi_z(p) = \log\pi_z(p)$ from \cite{hyperquadtree}, since they can map hyperbolic quadtree cells to nicer $(d-1)$-dimensional and one-dimensional Euclidean quadtree cells.

\section{Euclidean, spherical and small-scale hyperbolic geometry} \label{sec:doublingspace}
We first restrict to Euclidean, spherical and small-diameter hyperbolic geometry.
Intuitively, these are all ``close enough'' to Euclidean geometry to allow for the same general techniques (for example, they have bounded doubling dimension).

\subsection{Quadtrees}
As the basis for all our results, we use \emph{infinite quadtrees} tailored to specific geometries, defined as follows:
\begin{enumerate}[label=\textbf{QT\arabic*}]
    \item\label{it:qt} An infinite quadtree has some bounded region as its \emph{root} and some $L \in \mathbb Z$ as its highest level.
    For any given level $\ell \leq L$, the infinite quadtree subdivides the root into regions called \emph{cells}.
    Each cell at level $\ell$ is the disjoint union of its \emph{child cells} at level $\ell-1$, giving the quadtree its tree structure.
    \item\label{it:celldiam} A cell at level $\ell$ has diameter between $2^\ell / \sqrt{c_{\diam}d}$ and $2^\ell \cdot \sqrt{c_{\diam}d}$, where $c_{\diam}$ is a universal constant.
    \item\label{it:qtdegree}
    An infinite quadtree has \emph{bounded growth} if any cell has at most $2^{dk} \cdot d^{c_\text{growth} d}$ descendant cells $k$ levels below it, where $c_\text{growth}$ is a universal constant.
    \item\label{it:smallbcell} We call a cell $C$ a \emph{small-boundary cell} when for any $\eps \in (0,1]$ there is an $\eps\diam(C)$-covering on its boundary (i.e.\ a set of points such that the balls of radius $\eps\diam(C)$ around these points cover the boundary of $C$) of size at most $d^{c_{\text{cover}} d}/\eps^{d-1}$, where $c_{\text{cover}}$ is a universal constant.
\end{enumerate}

\paragraph{Construction.}
In the Euclidean case, we use the standard infinite quadtree.
For our purposes, we define it such that at any level $\ell \in \mathbb Z$ the cells form a $d$-dimensional grid of side length $2^\ell / \sqrt{d}$, aligned such that the origin is a grid vertex.
It is well-known that this has all described properties.
We also use the following lemma, proven in this form by Chan, Har-Peled and Jones \cite{locsens}.
\begin{lemma}[Lemma 3.7 of \cite{locsens}]\label{lem:locsensshift}
    Consider any two points $p,q \in [0,1)^d$, and let $\cQ$ be the infinite quadtree of $[0,2)^d$.
    For $D = 2\lceil d/2 \rceil$ and $i = 0, \dots, D$, let $v_i = (i / (D+1), \dots, i / (D+1))$.
    Then there exists an $i \in \{0, \dots, D\}$, such that $p + v_i$ and $q + v_i$ are contained in a cell of $\cQ$ with side length $2(D+1) \|p-q\|$.
\end{lemma}

In the hyperbolic case, we use the quadtree of \cite{hyperquadtree} described in Section~\ref{sec:prelims}.

For spherical geometry we define a new quadtree.
For that, first fix an arbitrary pair of antipodal points $o^+, o^-$ and consider the stereographic projection $\stereo$ from $o^-$ onto the tangent hyperplane of $o^+$.
Take a ball $B \subset \Sph^d$ around $o^+$ of radius $\frac{2}{3} \pi$ (note that this distance is measured along the sphere), then let $R$ be the minimum axis-aligned hypercube enclosing $\stereo(B)$, scaled up from its centre by a factor two for later purposes.
We get a quadtree for $\stereo^{-1}(R)$ by splitting up $R$ as a Euclidean quadtree, then projecting every cell back to $\Sph^d$ with $\stereo^{-1}$.
We will see that $\stereo^{-1}|_R$ only introduces a mild distortion, which means that all necessary properties still hold with slightly worse constants.
If we do the same construction with $o^+$ and $o^-$ swapped we get a second quadtree that together with the first covers the entire space.
Finally, we add a cell at the top level of both quadtrees that by itself covers the entire sphere $\Sph^d$.
This is necessary to ensure that any two points on the sphere will always share some quadtree cell, but the cell does violate \ref{it:qt} by not being the union of its children, so we will mention it as a special case where relevant.

\begin{figure}
    \centering
    \includegraphics{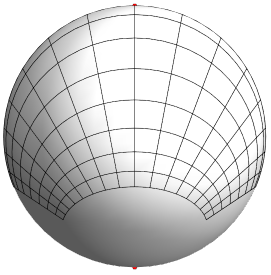}
    \caption{Grid on $\Sph^2$.}
    \label{fig:spheregrid}
\end{figure}

We remark that our quadtrees require some upper bound on the maximum distance $\Delta$ of point pairs for technical reasons, however, the value of this maximum distance does not appear in any of our size or running time bounds. In a practical setting one can typically derive such a $\Delta$ based on the representation of the input coordinates, even in a dynamic setting.

\begin{lemma}\label{lem:shift}
    Let $\mathbb{X}\in \{\Hyp,\Euc,\Sph\}$ and $X$ be a ball
    in $\mathbb{X}^d$.
    There is a set of $\bigO{d}$ infinite quadtrees such that for any two points $p,q \in X$, segment $pq$ is contained in a cell with diameter $\bigO{d \sqrt d} \cdot |pq|$ in one of the quadtrees.
    When $\mathbb{X} \in \{\Euc,\Sph\}$ or when $\mathbb{X}=\Hyp$ and
    the highest quadtree level is at most $4 + \log\log d$,
    these quadtrees have
    bounded growth
    and consist of small-boundary cells.
\end{lemma}
\begin{proof}
    In the Euclidean case, it is known that the quadtree has maximum degree $2^d$ and consists of convex cells.
    Each cell $C$ is a small-boundary cell, because for each boundary hyperplane the $(d-1)$-dimensional grid with $\frac{1}{\eps}$ points along each axis will give an $\eps\diam(C)$-covering.
    We can rescale $X$ so that it fits in the unit hypercube $[0,1)^d$; then we get the requested set of quadtrees from Lemma~\ref{lem:locsensshift}.
    \medskip

    In the hyperbolic case, we use the quadtree of \cite{hyperquadtree}.
    In particular, let $R'$ be the smallest hyperbolic quadtree cell that covers $X$ and let $R$ be a parent cell of $R'$ such that $R'$ is its child with minimal coordinates.
    We use the quadtree with root cell $R$ and the shifts of this quadtree given by \Cref{lem:hyperquadtree}\ref{lem:Hshift}; by construction each of these quadtrees will cover $X$.
    As given by \Cref{lem:hyperquadtree}\ref{lem:Hcelldiam}, cells of these quadtrees have diameter $\bigTh{2^\ell}$, which implies property \ref{it:celldiam}, so they are indeed valid quadtrees under our definition.
    Combining \Cref{lem:hyperquadtree}\ref{lem:Hshift} with \Cref{lem:hyperquadtree}\ref{lem:Hstarshaped} gives the lemma's main result.

    Furthermore, if $R$ has level at most $4 + \log\log d$, then $\diam(R) = \bigO{\log d}$.
    By \Cref{lem:hyperquadtree}\ref{lem:Hdegree}, $R$ has $\max\{2^d, d^{\bigO{d}}\}$ children and, importantly, a cell at $i$ levels below $R$ has $\max\{2^d, d^{\bigO{d / 2^i}}\}$ children.
    Thus, the number of descendants at $k$ levels below $R$ is (using $\exp_2(x)$ for $2^x$)
    \begin{align*}
        \prod_{i=0}^{k-1} \max\{2^d, d^{\bigO{d / 2^i}}\}
        &= \exp_2\left( \sum_{i=0}^{k-1} \max\{d, \bigO{d \log d} / 2^i\} \right) \\
        &\leq \exp_2( dk + \bigO{d \log d} ) \\
        &= 2^{dk} \cdot d^{\bigO d}.
    \end{align*}
    
    Additionally, the diameter bound lets us prove that the hyperbolic distance in $R$ only differs from an appropriately scaled Euclidean distance by a factor $d^{\bigO1}$.
    Given points $(x,z), (x',z') \in R$ at Euclidean distance $\delta_\Euc$, their hyperbolic distance is $\delta_\Hyp = 2\arsinh\left(\frac{\delta_\Euc}{2\sqrt{zz'}}\right)$.
    Thus, $\delta_\Hyp \leq \delta_\Euc / \zd(R)$ and
    \[ \delta_\Euc / \zu(R) \leq 2e^{\delta_\Hyp/2} \leq e^{\bigO{\log d}} \delta_\Hyp = d^{\bigO1} \delta_\Hyp. \]
    Seeing as also $\zu(R) / \zd(R) \leq 2^{\bigO{\log d}} = d^{\bigO1}$, this means that $\delta_\Hyp \leq \delta_\Euc / \zd(R) \leq d^{\bigO1} \delta_\Hyp$.
    Therefore we have distortion $d^{\bigO1}$ compared to $\delta_\Euc / \zd(R)$, which means that for some constant $c$ a Euclidean $\frac{\eps\diam(C)}{d^c}$-covering along the face will also be a hyperbolic $\eps\diam(C)$-cover.
    Such a cover has $\left( \eps / 2^{\bigO{\log d}} \right)^{1-d} = d^{\bigO{d}} \eps^{1-d}$ points, so we have small-boundary cells.
    \medskip

    For the spherical case, we first prove that the stereographic projection only introduces small distortion.
    \begin{claim*}
        For points $a,b \in R$ we have $\frac{1}{1 + 12d} |ab| \leq |\stereo^{-1}(a)\stereo^{-1}(b)| \leq |ab|$.
        If $a,b \in \stereo(B)$ this improves to $\frac{1}{4} |ab| \leq |\stereo^{-1}(a)\stereo^{-1}(b)| \leq |ab|$.
    \end{claim*}
    \begin{claimproof}
        Use coordinates $(p_1, \dots, p_{d+1}) \in \Euc^{d+1}$ with $\|p\| = 1$ for points $p \in \Sph^d$ and let $o^+ = (0, \dots, 0, -1)$.
        Then,
        $\stereo^{-1}(a) = \frac{4}{\|a\|^2 + 4} (a_1, \dots, a_d, \frac{1}{4}\|a\|^2 - 1)$.
        From this we can calculate that the chord between $\stereo^{-1}(a) \in \Sph^d$ and $\stereo^{-1}(b) \in \Sph^d$ has length $\frac{4|ab|}{\sqrt{(4 + \|a\|^2)(4 + \|b\|^2)}}$, which is a lower bound on the spherical distance $|\stereo^{-1}(a)\stereo^{-1}(b)|$.
        Here,
        $\stereo(B)$ has radius $2 \tan(\pi / 3) = 2\sqrt{3}$, which bounds $\|a\|$ and $\|b\|$ and thus implies $|\stereo^{-1}(a)\stereo^{-1}(b)| \geq \frac{1}{4} |ab|$ when $a,b \in \stereo(B)$.
        When only $a,b \in R$, then the largest norm we can get is $4\sqrt{3d}$.
        That directly gives $|\stereo^{-1}(a)\stereo^{-1}(b)| \geq \frac{1}{1 + 12d} \cdot |ab|$.

        For the upper bound, we calculate the spherical distance from the chord length as $|\stereo^{-1}(a)\stereo^{-1}(b)| = 2 \arcsin \frac{2|ab|}{\sqrt{(4 + \|a\|^2)(4 + \|b\|^2)}}$.
        Note that we get the largest spherical distance $|\stereo^{-1}(a)\stereo^{-1}(b)|$ for a given Euclidean distance $|ab|$ when $ab$ goes through the origin, since we can always move $ab$ towards the origin in a way that decreases both $\|a\|$ and $\|b\|$.
        For the upper bound, we can therefore assume $ab$ indeed goes through the origin $o = \stereo(o^+)$, which lets us split it into segments $ao$ and $ob$.
        Now,
        \[
            |\stereo^{-1}(a)o^+|
            = 2\arcsin\frac{\|a\|}{\sqrt{4 + \|a\|^2}}
            = 2\arctan\frac{\|a\|}{2}
            \leq \|a\|,
        \]
        which implies the same for $ob$ and thus means $|\stereo^{-1}(a)\stereo^{-1}(b)| \leq |ab|$.
    \end{claimproof}

    In the stereographic projection, a quadtree cell $C$ at level $\ell$ has $\diam(C) = 2^\ell \diam(R) = 2^\ell \cdot 4\sqrt{3d}$.
    Thus, $2^\ell \cdot \frac{4 \sqrt{3d}}{1 + 12d} \leq \diam(\stereo^{-1}(C)) \leq 2^\ell \cdot 4\sqrt{3d}$, which means that our spherical quadtree satisfies the diameter property.
    Additionally, we can get a $\frac{\eps \diam(C)}{1 + 12d}$-covering for $C$, which gets projected back to an $\eps \diam(C)$-covering of $\stereo^{-1}(C)$, meaning it is a small-boundary cell.
    The spherical quadtree inherits the maximum degree of the Euclidean quadtree, thus it satisfies all required properties.

    What remains is to find an appropriate set of these quadtrees for the lemma.
    For this, let $R'$ be $R$ scaled down by a factor two.
    Lemma~3.3 of \cite{Chan1998} now gives a set of $d+2$ shifts so that any point $p \in R'$ will have distance $2^\ell / (2d+4)$ from the cell boundary of a cell with side length $2^\ell$ in one of the shifted quadtrees that contains it.
    In other words, any ball of radius $r$ will be contained in a cell with diameter at most $r \cdot (4d+8)\sqrt{d}$ in some shifted quadtree.
    For any segment $pq \subset B$, consider the intersection $S \subseteq \Sph^d$ of $B$ and the ball with diameter $pq$.
    Since $S \subseteq B$ we have $\diam(\stereo(S)) \leq 4\diam(S) \leq 4|pq|$, so we can cover $\stereo(S)$ with a Euclidean ball of radius $r \leq 4|pq|$, which as mentioned will be contained in a cell with diameter at most $r \cdot (4d+8)\sqrt{d} = 4(4d+8)\sqrt{d} |pq|$ in some shifted quadtree.
    This proves the lemma for segments in the ball of radius $\frac{2}{3} \pi$ around $o^+$.
    To prove it for all of $\Sph^d$, we first repeat the quadtree construction but with $o^+$ and $o^-$ swapped.
    The lemma now holds for segments $pq$ where all points on $pq$ are within distance $\frac{2}{3}\pi$ from the same point $o^+$ or $o^-$.
    If there is a pair of points $a,b \in pq$ where $|ao^+| > \frac{2}{3}\pi$ while $|bo^-| > \frac{2}{3}\pi$, then necessarily $|ab| > \frac{1}{3} \pi$.
    Therefore also $|pq| > \frac{1}{3} \pi$ and we can take the cell that covers all of $\Sph^d$ (and has diameter $\pi$) to prove the lemma.
\end{proof}

\begin{lemma}\label{lem:quadtree_operations}
    These quadtrees support the following operations:
    \begin{enumerate}
        \item Splitting a cell into its children takes $\bigO{d}$ time per child.
        \item Constructing the covering for a small-boundary cell $C$ takes $d^{\bigO{d}}/\eps^{d-1}$ time.
        \item Given a pair of points $p,p'$ and an infinite quadtree, consider its smallest cell $C$ that contains $p$ and $p'$.
        \begin{enumerate}
            \item We can find $C$ in $\bigO{d}$ time.
            \item
            There is a fixed order on the children of $C$ where we can determine in $\bigO{d}$ time if the child containing $p$ comes before or after the child containing $q$.
        \end{enumerate}
    \end{enumerate}
\end{lemma}
\begin{proof}
    All results for Euclidean quadtrees also carry over to spherical quadtrees, since all operations can be done as a combination of the corresponding Euclidean operation with $\stereo$ and/or $\stereo^{-1}$.
    \begin{enumerate}
        \item
        In both Euclidean and hyperbolic space, each cell is described with a Cartesian product of $d$ intervals, each of which can be calculated in $\bigO{1}$ time from the parent cell's corresponding interval.
        
        \item
        In both the Euclidean and the hyperbolic case, Lemma~\ref{lem:shift} gave an $\eps\diam(C)$-covering using Euclidean grids.
        We can generate these using simple arithmetic in $\bigO{1}$ time per point.
        
        \item
        In Euclidean space, both statements are proven for example as Corollary~2.17 of \cite{GeomAppAlg}.
        In hyperbolic space, (b) is given by Lemma~\ref{lem:hyperquadtree}\ref{lem:Lorder} and (a) follows from a similar procedure.
        First, define $\tilde\pi(p)=(\tilde\pi_x(p), \tilde\pi_z(p))$; we will find $C$ using $(x,z)=\tilde\pi(p)$ and $(x',z')=\tilde\pi(p')$.
        If $\lfloor x/2^z \rfloor = \lfloor x'/2^{z'} \rfloor$ and $\lfloor z \rfloor = \lfloor z' \rfloor$ then $p$ and $p'$ lie in the same level-$0$ cell, which is split in a Euclidean way.
        Thus, we find the smallest cell containing $(x / 2^{\lfloor z \rfloor}, z)$ and $(x' / 2^{\lfloor z' \rfloor}, z')$ in a $d$-dimensional Euclidean quadtree.
        
        Otherwise, let $t = \lca(x, x', \tilde\pi_x(T^0))$ (which can be calculated from a $(d-1)$-dimensional Euclidean quadtree).
        Then, let $\ell$ be the first level where $\min\{z,z'\}$ and $\max\{t,z,z'\}$ are in the same cell of a one-dimensional Euclidean quadtree.
        Finally, $C$ is the cell at level $\ell$ containing $p$. \qedhere
    \end{enumerate}
\end{proof}

\subsection{Steiner spanner}
The key geometric lemma needed to reduce the number of Steiner points in the setting of $\Euc^d$, $\Sph^d$ and constant-diameter chunks of $\Hyp^d$ is the following lemma. The lemma is inspired by (and generalizes) Lemma~4.2 of \cite{LeS19}.

\split
\begin{proof}
    Consider the right-angled triangle $pxy$.
    In each case, we will prove that $|px| \leq (1 + \eps) |py|$ (or $|px| \leq (1 + 2^\Delta \eps) |py|$ in the hyperbolic case).
    We can prove $|qx| \leq (1 + \eps) |qy|$ (or $|qx| \leq (1 + 2^\Delta \eps) |qy|$) with the same reasoning and these two bounds jointly prove the lemma.

\begin{itemize}
    \item[$\Euc$.]
    By the Pythagorean theorem,
    \[ |px|^2 = |xy|^2 + |py|^2 \leq (1 + \eps) |py|^2. \]
    Thus $|px| \leq \sqrt{1 + \eps} \cdot |py| \leq (1 + \eps) |py|$.
    
    \item[$\Sph$.]
    The spherical law of cosines gives the following analogue to the Pythagorean theorem:
    \begin{equation}\label{eq:sphericalpyth}
        \cos|px| = \cos|xy| \cdot \cos|py|.
    \end{equation}
    We will first assume that $|py| < \pi/2$ so that $\cos|py| > 0$ and every instance of $\cos$ has its argument in the range $[0,\pi]$, making it strictly decreasing.
    Thus, it will be enough to prove that $\cos|px| \geq \cos((1+\eps)|py|)$.
    Here, $\cos|px| = \cos|xy| \cdot \cos|py|$ by Equation~\ref{eq:sphericalpyth} and by cosine's angle sum identity
    \[\cos((1+\eps)|py|) = \cos|py| \cdot \cos(\eps|py|) - \sin|py| \cdot \sin(\eps|py|).\]
    We divide both by $\cos|py|$ so that the inequality to be proven becomes
    \begin{equation}\label{eq:sphericaltoshow}
        \cos|xy| \geq \cos(\eps|py|) - \tan|py| \cdot \sin(\eps|py|).
    \end{equation}
    On the left side, we use that $\cos t \geq 1 - \frac{t^2}{2}$ for $t \in \Reals$ to get $\cos|xy| \geq \cos(\sqrt{\eps}|py|) \geq 1 - \frac{\eps|py|^2}{2}$.
    On the right side, we use that $\tan t \geq t$, that $\sin t \geq \frac{2}{\pi} \cdot t$ and that $\cos t \leq 1$, for $t \in [0, \pi)$.
    This gives
    \[\cos(\eps|py|) - \tan|py| \cdot \sin(\eps|py|) \leq 1 - \frac{2\eps|py|^2}{\pi}.\]
    Therefore, Equation~\ref{eq:sphericaltoshow} is implied by $1 - \frac{\eps|py|^2}{2} \geq 1 - \frac{2\eps|py|^2}{\pi}$, which holds because $\frac{1}{2} \leq \frac{2}{\pi}$.
    
    The case where $|py| \geq \pi/2$ is straightforward:
    if $(1+\eps) |py| > \pi$ then $|px| \leq \pi < (1+\eps) |py|$, while in case of $(1+\eps) |py| \leq \pi$ we have that $\cos$ is strictly decreasing, so we directly get
    \[\cos((1+\eps) |py|) \leq \cos|py| \leq \cos|py| \cdot \cos|xy| = \cos|px|,\]
    which also proves $|px| \leq (1+\eps)|py|$.
    
    \item[$\Hyp$.]
    The hyperbolic analogue to the Pythagorean theorem yields
    \begin{equation}\label{eq:hyperpyth}
         \cosh|px| = \cosh|xy| \cdot \cosh|py|.
    \end{equation}
    We will use a similar approach to the spherical case, except now we work towards $\cosh|px| \leq \cosh((1 + c\eps) |py|)$ with $c = \frac{e^\Delta}{2\Delta}$ and $|py| \leq \Delta$.
    Here, $\cosh|px| = \cosh|xy| \cdot \cosh|py|$ by Equation~\ref{eq:hyperpyth} and the sum of arguments identity of $\cosh$ gives \[\cosh((1 + c\eps) |py|) = \cosh|py| \cdot \cosh(c\eps |py|) + \sinh|py| \cdot \sinh(c\eps|py|).\]
    We divide both by $\cosh|py|$ so that the inequality to be proven becomes
    \begin{equation}\label{eq:hypertoshow}
        \cosh|xy| \leq \cosh(c\eps |py|) + \tanh|py| \cdot \sinh(c\eps|py|).
    \end{equation}
    On the left side, we use that $\cosh t \leq 1 + \frac{\cosh\Delta - 1}{\Delta^2} \cdot t^2$ for $t \in [0, \Delta]$ to get \[\cosh|xy| \leq \cosh(\sqrt{\eps}|py|) \leq 1 + \frac{\cosh\Delta - 1}{\Delta^2} \cdot \eps|py|^2.\]
    On the right side, we use that $\sinh t \geq t$ and $\cosh t \geq 1$ for all $t \geq 0$ and $\tanh t \geq \frac{\tanh\Delta}{\Delta} \cdot t$ for $t \in [0, \Delta]$.
    This gives
    \[\cosh(c\eps |py|) + \tanh|py| \cdot \sinh(c\eps|py|) \geq 1 + \frac{\tanh\Delta}{\Delta} \cdot c\eps|py|^2.\]
    Therefore, Equation~\ref{eq:hypertoshow} is implied by $1 + \frac{\cosh\Delta - 1}{\Delta^2} \cdot \eps|py|^2 \leq 1 + \frac{\tanh\Delta}{\Delta} \cdot c\eps|py|^2$, which holds because $\frac{\cosh(\Delta) - 1}{\Delta \tanh\Delta} \leq \frac{e^\Delta}{2\Delta} = c$.
    \qedhere
\end{itemize}
\end{proof}

\paragraph{Constructing a bipartite Steiner spanner $G_1(\eps,P)$.}
Given a cell $C$, we define $\desc(C, a)$ as the set of descendants of $C$ in its quadtree at the highest level where all their diameters are at most $a \diam(C)$.
Let $P_i$ denote the list of points in $P$ sorted by the order corresponding to the $i^\text{th}$ quadtree from Lemma~\ref{lem:shift}.
We now define the graph $G_1(\eps,P)$ that we will later prove to be a Steiner spanner.
For any pair of points $p,q$ adjacent in some $P_i$, we will consider each cell $C$ of quadtree $i$ where $p,q \in C$ and $\diam(C) \leq \frac{1}{\eps} |pq|$.
For each such cell $C$ we will construct a set $C_S$ of Steiner points.
First consider $\cC = \desc(C, \frac1{6c_{\text{shift}}d \sqrt d})$, where $c_{\text{shift}}$ is the constant hidden by the big-O notation in Lemma~\ref{lem:shift}.
Here, cells in $\cC$ lie at most $\log(6c_{\diam}d \cdot c_{\text{shift}}d \sqrt d) = \bigO{\log d}$ levels below $C$, thus $|\cC| = d^{\bigO{d}}$ by \ref{it:qtdegree}.
For each cell $C' \in \cC$, we define a set of Steiner points $\cov(C')$ according to a $\sqrt\eps\diam(C')$-covering of its boundary, then let $C_S = \bigcup_{C' \in \cC} \cov(C')$.
Because each such cell $C' \in \cC$ is a small-boundary cell, $|C_S| = |\cC| \cdot |\cov(C')| = d^{\bigO{d}} \cdot \eps^{(1-d)/2}$.
For each considered cell $C$, we add $C_S$ to $G_1(\eps,P)$ as well as an edge from $p$ and $q$ to each $s \in C_S$.
By construction, $G_1(\eps,P)$ is bipartite and each input point has degree $d^{\bigO{d}} \cdot \eps^{(1-d)/2} \log \frac{1}{\eps}$.

\begin{lemma}\label{lem:lowdiam_implicit}
    There is a data structure of size $\bigO{dn}$ that reports in $d^{\bigO{d}} \cdot \eps^{(1-d)/2} \log \frac{1}{\eps} + \bigO{d^2 \log n}$ time the edges that are added or removed for $G_1(\eps,P)$ when a point in $P$ is added or removed.
\end{lemma}
\begin{proof}
    As data structure, we maintain the points $P$ in sorted order for each order given by one of the $\bigO{d}$ quadtrees of \Cref{lem:shift}.
    This can for example be done with a self-balancing binary search tree.
    To use $\bigO{dn}$ instead of $\bigO{d^2n}$ storage, we only store each point's coordinates once and then refer to it with pointers.
    The method of reporting changes now follows fairly directly from the construction of $G_1(\eps,P)$.
    Adding points to or removing them from the $\bigO{d}$ sorted sets takes $\bigO{d^2 \log n}$ time, as \Cref{lem:quadtree_operations} lets us compare points in $\bigO{d}$ time.
    If such an update causes points $p$ and $q$ to become adjacent in order $i$, then we can use the operations of \Cref{lem:quadtree_operations} to follow the construction of $G_1(\eps,P)$ and find the edges this adds in $d^{\bigO{d}} \cdot \eps^{(1-d)/2} \log \frac{1}{\eps}$ time.
    If an update causes some points $p$ and $q$ to no longer be adjacent, we can do the same operations to find the edges that should be removed.
\end{proof}

We can explicitly maintain $G_1(\eps,P)$ by maintaining the data structure of Lemma~\ref{lem:lowdiam_implicit} alongside an adjacency list representation of the graph.
Adding an edge or Steiner point to this representation takes $\bigO{d \log n}$ time, so this gives Theorem~\ref{thm:doublingmain}.
To prove that $G_1(\eps,P)$ is indeed a Steiner $(1+\eps)$-spanner, we first separate out \Cref{lem:reprapprox}, which will also be used later.

\begin{lemma}\label{lem:reprapprox}
    Let $G$ be a geometric graph on $P \cup S \subset \Hyp^d$, let $\eps \in (0,\frac{1}{2}]$ and let $c > 0$.
    Assume that for any pair $p,q \in P$ there are points $p',q' \in P$ with $|pp'| \leq \frac{\eps}{c+3}|pq|$ and $|qq'| \leq \frac{\eps}{c+3}|pq|$ that are connected with a path of length at most $(1+\frac{c\eps}{c+3})|pq|$ in $G$.
    Then, $G$ is a Steiner $(1+\eps)$-spanner for $P$.
\end{lemma}
\begin{proof}
    This is a common result also used for the proofs of e.g.\ \cite[Theorem 3.12]{GeomAppAlg} and \cite[Lemma 20]{hyperquadtree}.
    We prove it by induction.
    Assume that we have proven $\dist_G(v,w) \leq (1 + \eps) |vw|$ for any pair $v,w \in P$ with $|vw| < |pq|$.
    As base case, this holds trivially for $|vw| = 0$.
    Now we will prove it for some pair $p,q \in P$.
    Take $p',q' \in P$ according to the lemma statement.
    From the induction hypothesis, $\dist_G(p,p') \leq (1 + \eps) |pp'|$, which we can further bound by $(1+\eps) \cdot \frac{\eps}{c+3}|pq| \leq \frac{3}{2} \cdot \frac{\eps}{c+3} |pq|$ using our assumptions.
    Analogously, $\dist_G(q,q') \leq \frac{3}{2} \cdot \frac{\eps}{c+3} |pq|$
    and therefore
    \[\dist_G(p,q) \leq \dist_G(p,p') + \dist_G(p',q') + \dist_G(q',q) \leq (1 + \eps) |pq|,\]
    concluding the proof.
\end{proof}

\begin{theorem}\label{thm:doublingmain}
    Let $\eps \in (0,\frac12]$ and $P$ be a set of $n$ points in $\Euc^d$ or $\Sph^d$, or in a
    quadtree cell in $\Hyp^d$ of level $4+\log\log d$.
    There is a Steiner $(1+\eps)$-spanner for~$P$ with Steiner vertex set $Q$, that is bipartite on $P$ and $Q$, and where each point $p \in P$ has degree $d^{\bigO{d}} \cdot \eps^{(1-d)/2} \log \frac{1}{\eps}$.
    We can maintain it dynamically such that each point insertion or deletion takes $d^{\bigO{d}} \cdot \eps^{(1-d)/2} \log \frac{1}{\eps} \cdot \log n$ time.
\end{theorem}
\begin{proof}
    For this, we will prove that $G_1(P,\eps')$ is a Steiner $(1+\eps)$-spanner, when $\eps' = \eps / (c + 3)$ and $c = e^{16 \log d} = d^{16 \log e}$.
    The edge count and dynamic maintenance of $G_1(P,\eps')$ have already been proven; left to show is that we indeed have a Steiner spanner.
    Let $\dist_G$ denote the distance in $G_1(P,\eps')$.
    Given $p,q \in P$ we want to prove $\dist_G(p,q) \leq (1 + (c+3)\eps') |pq|$.
    Let $C$ be the smallest cell among the quadtrees where $p,q \in C$; by Lemma~\ref{lem:shift} we know $\diam(C) \leq c_{\text{shift}}d \sqrt d |pq|$ for some constant $c_{\text{shift}}$.
    As in the construction, let $\cC = \desc(C, \frac1{c_{\text{shift}}d \sqrt d})$.
    Thus, cells $C' \in \cC$ will have $\diam(C') \leq \frac{1}{6}|pq|$.
    The midpoint $m$ of $pq$ lies in such a cell $C' \in \cC$, and because $\diam(C') \leq \frac{1}{6}|pq|$ there will be a Steiner point $s \in \cov(C')$ that is both within distance $\frac{1}{6}|pq|$ from $m$ and within distance $\frac{1}{6}\sqrt{\eps'}|pq|$ of $pq$ itself.
    The distance from $m$ to the point $y \in pq$ closest to $s$ is at most $|ms| + |sy| \leq \frac{1}{3}|pq|$, meaning $\min\{|py|, |qy|\} \geq \frac{1}{6}|py|$.
    Thus, Lemma~\ref{lem:split} gives us that $|ps| + |sq| \leq (1 + c\eps')|pq|$.
    
    Now, let $p' \in P$ be the point in the same cell from $\cC_{\eps'} = \desc(C, \frac{\eps'}{c_{\text{shift}}d \sqrt d})$ as $p$ closest in the order to $q$, and $q'$ defined analogously.
    If one of $p$ and $q$ is not contained in any cell from $\cC_{\eps'}$ (because $C$ is the special top cell for $\Sph^d$), then assume without loss of generality that this is $q$ and set $q'=q$ instead.
    By construction, $|pp'| \leq \frac{\eps'}{c_{\text{shift}}d \sqrt d} \cdot \diam(C) \leq \eps' |pq|$.
    Additionally, $p'$ will be adjacent in the order to some point in a different cell of $\cC_{\eps'}$, which means it is connected to the Steiner points of $C$ and in particular to $s$.
    These same things hold for $|qq'|$ and $q$.
    Thus, $\dist_G(p',q') \leq |p'p| + |ps| + |sq| + |sq'| \leq (1 + 2\eps' + c\eps')|pq|$.
    This means we exactly satisfy the conditions of Lemma~\ref{lem:reprapprox}, making this a Steiner $(1 + \eps)$-spanner.
\end{proof}

As it turns out, the Steiner points in this Steiner spanner also have small degree, which gives a straightforward way to turn it into a spanner without Steiner points.

\begin{theorem}\label{thm:doublingspanner}
    Let $\eps \in (0,\frac12]$ and $P$ be a set of $n$ points in $\Euc^d$ or $\Sph^d$, or in a
    quadtree cell in $\Hyp^d$ of level $4+\log\log d$.
    In $d^{\bigO{d}} \cdot \eps^{-d} \log\frac1\eps \cdot n\log n$ time, we can construct a $(1+\eps)$-spanner (i.e., without Steiner points) for $P$ with maximum degree $d^{\bigO{d}} \cdot \eps^{-d} \log\frac1\eps$.
    \footnote{Note that we do not get optimal edge count $\bigOd{n / \eps^{d-1}}$, so one can likely improve on this by adapting for example \cite{GaoH24}.}
\end{theorem}
\begin{proof}
    To construct the $(1+\eps)$-spanner, start with the Steiner $(1+\eps)$-spanner of Theorem~\ref{thm:doublingmain}.
    We form the $(1+\eps)$-spanner by connecting each point $p \in P$ to the neighbours of the Steiner points $p$ is connected to in the Steiner $(1+\eps)$-spanner.
    It remains to prove that this gives the prescribed degree bound.
    Based on the construction, the Steiner points connected to $p$ come from $\bigO{d\log\frac{d}{\eps}}$ sets $C_S$, where Steiner points from the same set $C_S$ have the same neighbours.
    In particular, they are connected to at most one point for each descendant cell of $C$, and only if that descendant cell has diameter at least $\eps\diam(C)$.
    By \ref{it:qtdegree}, this means each Steiner point has degree $d^{\bigO{d}} / \eps^d$.
    Thus in total, the degree of $p$ is at most $d^{\bigO{d}} \cdot \eps^{-d} \log\frac1\eps$.
\end{proof}

\section{Large-scale hyperbolic geometry}\label{sec:hyperbolic}
We will now prove a similar Steiner spanner result for large distances in hyperbolic geometry.
Here, the exponential expansion of hyperbolic space becomes a problem and we need to exploit its tree-likeness, while still covering cell boundaries with Steiner points as in the construction of Section~\ref{sec:doublingspace}.
This gives Theorem~\ref{thm:Hbipartite} as the first result, which Lemma~\ref{lem:bipartitelowerbound} in fact proves near-optimal if we require the Steiner spanner to be bipartite.

However, it does not yet match the bounds of Theorem~\ref{thm:main} for $d=2$.
To further reduce the number of edges, we need to connect Steiner points to Steiner points, allowing the input points to have a smaller degree.
This finally gives Theorem~\ref{thm:Hmain}, which matches the Euclidean bounds and goes far below them for large distances.
Using the same ideas, we can also construct a Steiner spanner with arbitrarily small \emph{additive} error, which we do in Section~\ref{sec:additivespanner}.

\subsection{Bipartite Steiner spanner} \label{sec:bipartitespanner}
We first use the same general framework as in Section~\ref{sec:doublingspace}, but we will now also need to properly deal with the exponential expansion of hyperbolic space.
To do so, we will introduce alternatives to the shifting result of \Cref{lem:shift} and the covering result of \Cref{lem:split} that are more suitable for this setting.
To prove this alternative shifting result, we first need the following lemma:
\begin{lemma}\label{lem:lcadiff}
    Let $p,q \in \Hyp^d$ and $t$ be the highest point on $pq$.
    Then, the cell $C_t = \cell(t, T^0)$ has $\diam(\pi_x(C_t)) > \|x(p) - x(q)\| / 4$.
\end{lemma}
\begin{proof}
    Since we are at level $0$, we have $w(C_t) = 1/\sqrt{d-1}$, meaning $\pi_x(C_t)$ is a $(d-1)$-dimensional Euclidean hypercube with side length $\zd(C_t)/\sqrt{d-1}$ and thus $\diam(\pi_x(C_t)) = \zd(C_t)$.
    Because $C_t$ must contain $t$, here $\zu(C_t) \geq z(t)$ which implies $\zd(C_t) \geq z(t) / 2$ as $h(C_t) = \log\frac{\zu(C_t)}{\zd(C_t)} = 1$ at level $0$.
    
    We will now bound $z(t)$ in terms of $\|x(p) - x(q)\|$.
    In the halfspace model, $pq$ is either a vertical Euclidean segment or an arc from a Euclidean circle $S$ with its origin at $z=0$.
    In the first case, $\|x(p) - x(q)\| = 0$ so the result holds trivially.
    Otherwise, first assume $t$ is distinct from $p$ and $q$.
    Then, $t$ must be the highest point of $S$, meaning $S$ has radius $z(t)$.
    Since $x(p)$ and $x(q)$ must lie in the interior of $\pi_x(S)$, now $\|x(p) - x(q)\| < 2z(t)$ which implies $\diam(\pi_x(C_t)) > \|x(p) - x(q)\| / 4$.
    Finally, consider $t=p$; the case $t=q$ is symmetric.
    Consider the distance from $x(p)$ to the closest boundary point $b$ of $\pi_x(S)$.
    The points $a \in \Hyp^d$ where $\|x(a) - b\| \leq z(a)$ are exactly those that lie above the Euclidean line through $(b, 0)$ and the highest point on $S$.
    In particular, this means $\|x(p) - b\| \leq z(p)$.
    Since $x(q)$ must be closer to $x(p)$ than $b$ is, $\|x(p) - x(q)\| < z(p)$, which implies $\diam(\pi_x(C_t)) > \|x(p) - x(q)\| / 2$.
\end{proof}
We also need the following lemma, which was used for the shifting result of \cite{hyperquadtree}.
\begin{lemma}[Lemma 9 of \cite{hyperquadtree}]\label{lem:1Dshift}
    Let $\cQ$ be a one-dimensional Euclidean quadtree whose largest cell is $[0,2)$.
    For any two points $p,q \in [0,1)$, there is a shift $\sigma \in \{0, \frac13, \frac23 \}$ such that when added to the quadtree $p + \sigma$ and $q + \sigma$ are contained in a cell of $\cQ$ with length $< 3|p-q|$ and one of the points is in the lower $\frac13$ of the cell.
\end{lemma}

One of the geometric insights behind our construction is captured by the following lemma, which essentially replaces Lemma~\ref{lem:shift}.
Here, we will start using the arborescences $T^\ell$ defined in \Cref{sec:prelims}.
In particular, we introduce $T^\ell_i$, which is the version of $T^\ell$ where all cells are shifted by the same transformation as the $i^\text{th}$ quadtree given by \Cref{lem:hyperquadtree}\ref{lem:Hshift}.
For a cell $C$ and $I \subseteq [0,1]$, we also define $C|_I$ as $\{ p \in C \mid \frac{\log z(p) - \log\zd(C)}{h(C)} \in I \}$, i.e.\ $C$ restricted to the points whose normalised distance from the bottom boundary falls in $I$.
\begin{restatable}{lemma}{lcashift}\label{lem:lca}
    Given a level $\ell \geq 4 + \log\log d$.
    For any two points $p,q \in \Hyp^d$ (with $|pq| \leq \Delta$), there is a shift~$i$ where $pq$ intersects the cell $C = \lca(p, q, T_i^\ell)$.
    In particular, $pq$ will intersect $C|_{[\frac{1}{3},1]}$.
\end{restatable}
\begin{proof}
    We use the same shifts as in \Cref{lem:hyperquadtree}\ref{lem:Hshift} (and thereby \Cref{lem:shift}).
    These are of the form $(x,z) \mapsto (\sigma x + \tau, \sigma z)$; a `vertical' shift that multiplies all coordinates by $\sigma$ followed by a `horizontal' shift that adds $\tau$ to the $x$-coordinates.
    In particular, we have all combinations of $\sigma = 2^{H \cdot i / 3}$ for $i \in \{0,1,2\}$ (the shifts from \Cref{lem:1Dshift}) and $\tau = \left( \frac{W \cdot j}{D+1}, \dots, \frac{W \cdot j}{D+1} \right)$ for $j \in \{0, \dots, D\}$ (the shifts from \Cref{lem:locsensshift}).

    First, let $t$ be the highest point on geodesic $pq$.
    According to \Cref{lem:1Dshift} we can take $\sigma$ so that $t$ lies in the middle $\frac{1}{3}$ of its cell, i.e.\ $t \in \cell(t, T^\ell_{\sigma,\tau})|_{[\frac{1}{3},\frac{2}{3}]}$ for any $\tau$.
    Now, project $p,q$ and each $T_{\sigma,\tau}^0$ onto the $x$-coordinate, so that we get $(d-1)$-dimensional points and $(d-1)$-dimensional infinite Euclidean quadtrees.
    By \Cref{lem:locsensshift} there must be a shift $\tau$ where the smallest cell $C \in T_{\sigma,\tau}^0$ such that $\pi_x(C)$ contains both $x(p)$ and $x(q)$ has $\diam(\pi_x(C)) \leq 2(d+1)\sqrt{d-1} \cdot \|x(p) - x(q)\|$.
    
    In $T_{\sigma,\tau}^0$, the lowest common ancestor of $p$ and $q$ is then either $C$ or a cell containing $p$ or $q$.
    In the latter case, we already know that $pq$ intersects $C$.
    We have also already ensured that $t$ must lie in the middle $\frac{1}{3}$ of the level-$\ell$ cell containing $C$, which proves the statement.
    
    Now, consider the case where $C = \lca(p, q, T_i^0)$ and let $C_t$ be the cell of $T_{\sigma,\tau}^0$ containing $t$.
    By Lemma~\ref{lem:lcadiff}, $\diam(\pi_x(C_t)) > \|x(p) - x(q)\| / 4$.
    Thus, $\diam(\pi_x(C)) / \diam(\pi_x(C_t)) < 8(d+1)\sqrt{d-1}$ which means they differ by less than $3 + \frac{3}{2}\log(d+1)$ levels of $T^0$.
    Going up these $3 + \frac{3}{2}\log(d+1)$ level-$0$ cells means we cross $2^{-\ell} \cdot (3 + \frac{3}{2}\log(d+1)) \leq \frac{1}{3}$ of the height of a level-$\ell$ cell, so because $t$ was in the middle $\frac{1}{3}$ of its level-$\ell$ cell it means we remain in that same cell.
    We also already know that $t$ lies in the upper $\frac{2}{3}$ of the level-$\ell$ cell.
\end{proof}

Note that when combined with a transitive closure spanner as in \cite{additivespanner}, Lemma~\ref{lem:lca} immediately gives a Steiner spanner with (for example) $\bigO{dn \log^* n}$ edges and additive error $\bigO{\log d}$.
However, we need error $\eps |pq|$, which can be much smaller than $\log d$.
To accomplish that, we will use groups of equally-spaced Steiner points, as in Section~\ref{sec:doublingspace}.
We first give an alternative to Lemma~\ref{lem:split} that works for larger hyperbolic distances and gives an additive error:
\Hsplit
\begin{proof}
    Let $y$ be the point on $pq$ closest to $x$.
    If $\delta \geq \sqrt{2}$ then the lemma already holds by the triangle inequality, so we can assume that $\delta < \sqrt{2}$.
    Now, again by the triangle inequality, $|py| \geq |px| - \delta > 2 - \sqrt{2} > \frac{1}{2} \ln 3$; we will use this bound later.
    First, we use the hyperbolic analogue to the Pythagorean theorem and that $\cosh t \leq e^{t^2 / 2}$ for $t \in \Reals$, giving
    \[
        \cosh|px| = \cosh\delta \cdot \cosh|py| \leq e^{\delta^2/2} \cosh|py|.
    \]
It remains to show that 
\begin{equation}\label{eq:toshow}
    e^{\delta^2/2} \cosh|py| \leq \cosh(\delta^2 + |py|),
\end{equation}
as $\cosh|px|\leq \cosh(\delta^2 + |py|)$ implies $|px|\leq \delta^2 + |py|$ because $\cosh$ is monotone increasing in the non-negative domain.

Equation~\eqref{eq:toshow} would follow if we can show $e^x \cosh a \leq \cosh(2x + a)$ for all $x \in \Reals$ and $a \geq \frac{1}{2} \ln 3$; let us prove this next.
    Note that $1 - e^{-3x}$ is concave while $3(e^x - 1)$ is convex, and both have the tangent line $y=3x$ at $x=0$.  Therefore,
 \[1 - e^{-3x} \leq 3x \leq 3(e^x - 1) \leq e^{2a} (e^x - 1)\]
After multiplying by $e^{x-a}$ we get:
    \begin{align*}
        e^{x-a} - e^{-2x-a} &\leq e^{2x+a} - e^{x+a} \\
        e^{x-a} + e^{x+a} &\leq e^{2x+a} + e^{-2x-a} \\
        e^x \cosh a &\leq \cosh(2x + a),
    \end{align*}
    which concludes the proof.
\end{proof}

\paragraph{Construction of a hyperbolic bipartite Steiner spanner $G_2(P,\eps)$.}
We now have all we need to construct a bipartite graph $G_2(P,\eps)$ and prove that it is a Steiner spanner.
To deal with small distances, we make use of Theorem~\ref{thm:doublingmain}:
for each shift of Lemma~\ref{lem:shift}, we partition the points $P$ based on level $4 + \log\log d$ of the quadtree, then follow the construction of Theorem~\ref{thm:doublingmain} for each of these sets.

For larger distances, we first introduce some new notation:
for any cell $C$, we let $T_i^\ell \cap C$ give the subgraph of $T_i^\ell$ induced by the cells of $T_i^\ell$ intersecting $C$.
Now consider each pair $p,q \in P$ adjacent in the order given by the $i^\text{th}$ quadtree from Lemma~\ref{lem:shift}.
Consider each cell $C$ of quadtree $i$ with $p,q \in C$ and $\diam(C) \leq \frac{1}{\eps} |pq|$, together with each shift $j$ given by \Cref{lem:lca}.
Depending on which is larger, let $\ell$ either be $4 + \log\log d$, or such that level $\ell$ is the first level where cells have diameter at most $\eps\diam(C)$; we will connect $p$ and $q$ to their ancestors in $C \cap T_j^\ell$.
More specifically, for each such ancestor $C'$ we define Steiner points $\cov(C')$ placed to form a $\sqrt{\eps\diam(C)}$-covering of $C'|_{\{\frac{1}{6}\}}$.
If $\diam(C') \leq \eps \diam(C)$ then $\cov(C')$ is a single point inside $C'$.
Otherwise, $\ell = 4 + \log\log d$ so $C'$ is a small-boundary cell and $|\cov(C')| = d^{\bigO{d}} \cdot (\eps \diam(C))^{(d-1)/2}$.
For each such cell $C'$, we add the Steiner points $\cov(C')$ to $G_2(P,\eps)$, as well as edges from $p$ and $q$ to these points.
That gives the following (intermediate) result, where $\pi_{[a,b]}$ clips its argument into interval~$[a,b]$, i.e.\ $\pi_{[a,b]}(x) := \max\{a, \min\{b, x\}\}$.

\begin{lemma}
    Let $\eps \in (0,\frac12]$ and $P$ be a set of $n$ points in $\Hyp^d$.
    There is a Steiner $(1+\eps)$-spanner for~$P$ with Steiner vertex set $Q$ that is bipartite on $P$ and $Q$, where each point $p \in P$ has degree $d^{\bigO{d}} \cdot \delta_p (\eps\delta_p)^{(1-d)/2} \log\frac{1}{\eps}$.
    Here, $\delta_p = \pi_{[1,\frac1\eps]}( \max_{q \in P} \dist(p,q) )$ for $d = 2$ and $\delta_p = \pi_{[1,\frac1\eps]}( \min_{q \in P \setminus\{p\}} \dist(p,q) )$ otherwise.
\end{lemma}
\begin{proof}
    We will prove that $G_2(P,\eps')$ gives the claimed results when $\eps' = \eps / (c+3)$ and $c = 2 + c_{\text{shift}} d\sqrt{d}$.
    It follows directly from the construction that $G_2(P, \eps')$ is bipartite on $P$ and $Q$ and that each point in $P$ has the given degree.
    It remains to show that $G_2(P, \eps')$ is a Steiner $(1+\eps)$-spanner; this proof is similar to that of \Cref{thm:doublingmain}.
    Let $p, q \in P$ and $C$ be the smallest cell among the quadtrees where $p,q \in C$; by Lemma~\ref{lem:shift} we know $\diam(C) \leq c_{\text{shift}} d\sqrt{d} |pq|$ for some constant $c_{\text{shift}}$.
    Then, let $p' \in P$ be the point in the same cell from $\desc(C, \frac{\eps'}{c_{\text{shift}}d \sqrt d})$ as $p$ closest in the order to $q$, and $q'$ defined analogously.
    Let $\ell$ be the level as defined in the construction.
    Now, \Cref{lem:lca} gives a shift $j$ where $p'q'$ intersects $C' = \lca(p',q',T_j^\ell)$.
    From \Cref{lem:shift} we know that $p'q' \subset C$ and thus $C' \cap C \neq \emptyset$, which means that $C'$ also exists in $C \cap T_j^\ell$.
    Thus, the construction ensures that both $p'$ and $q'$ are connected to $\cov(C')$.
    If $\ell$ was such that $\diam(C') \leq \eps'\diam(C)$ then there is a single Steiner point $s \in \cov(C')$ and $|p's| + |sq'| \leq |p'q'| + 2\diam(C') \leq |p'q'| + c_{\text{shift}} d\sqrt{d}\eps'|pq|$.
    
    Otherwise, $\cov(C')$ is a $\sqrt{\eps'\diam(C)}$-covering on $C'|_{\{\frac{1}{6}\}}$.
    Lemma~\ref{lem:lca} implies that $p'q'$ intersects $C'|_{[\frac{1}{3},1]}$, so in particular it will also intersect $C'|_{\{\frac{1}{6}\}}$ and both $p'$ and $q'$ will have distance at least $\frac{1}{6} h(C') \ln2 \geq 2$ from that set.
    By construction, $p'q'$ has distance at most $\sqrt{\eps' \diam(C)}$ from some Steiner point $s \in \cov(C')$.
    Lemma~\ref{lem:Hsplit2} now implies that also in this case $|p's| + |sq'| \leq |p'q'| + \eps'\diam(C) \leq |p'q'| + c_{\text{shift}} d\sqrt{d}\eps'|pq|$.
    
    Additionally, $p$ and $p'$ share a cell of diameter at most $\frac{\eps'}{c_{\text{shift}}d \sqrt d} \cdot \diam(C) \leq \eps'|pq|$ and so do $q$ and $q'$.
    In particular this means that $|p'q'| \leq |pp'| + |pq| + |qq'| \leq (1 + 2\eps') |pq|$ and thus also $|p's| + |sq'| \leq (1 + c\eps') |pq|$.
    We now satisfy the conditions of Lemma~\ref{lem:reprapprox} to make $G_2(P,\eps')$ a Steiner $(1 + \eps)$-spanner.
\end{proof}

Perhaps worth noting is that for large distances ($\delta_p = 1/\eps$) the construction becomes similar to that of \cite{hyperquadtree} and consequently we get the same bound; our main improvement is for small and intermediate distances.
As in \Cref{sec:doublingspace}, we can maintain this Steiner spanner very efficiently.
\begin{lemma}\label{lem:bipartite_implicit}
    There is a data structure of size $\bigO{dn}$ that reports in $d^{\bigO{d}} \cdot \delta_p (\eps\delta_p)^{(1-d)/2} \log\frac{1}{\eps} + \bigO{d^2 \log n}$ time the edges that are added or removed for $G_2(P,\eps)$ when a point is added or removed.
\end{lemma}
\begin{proof}
    We use the same data structure as in \Cref{lem:lowdiam_implicit}.
    Since the edges for each point are defined purely based on its neighbours in the quadtree orders, we can follow the given construction and only need to (un)do this for $\bigO{d}$ points when an insertion or deletion happens.
\end{proof}

Putting this together gives \Cref{thm:Hbipartite}, where we use the bound $\delta_p \leq 1$ for simplicity.

\begin{restatable}{theorem}{Hbipartite}\label{thm:Hbipartite}
    Let $\eps \in (0,\frac12]$ and $P$ be a set of $n$ points in $\Hyp^d$.
    There is a Steiner $(1+\eps)$-spanner for~$P$ with Steiner vertex set $Q$ that is bipartite on $P$ and $Q$.
    Each point in $P$ has degree at most $d^{\bigO{d}} \cdot \eps^{(1-d)/2} \log\frac{1}{\eps}$ when $d \geq 3$ and at most $\bigO{\frac{1}{\eps} \log\frac{1}{\eps}}$ when $d=2$.
\end{restatable}

Following the same steps as Corollary~16 of \cite{hyperquadtree}, \Cref{thm:Hbipartite} implies a $(2+\eps)$-spanner with the same edge bound.
We rephrase that corollary in the following theorem:
\begin{theorem}\label{thm:2spanner}
    Given a point set $P$ in some metric space $(X,\dist)$ and a bipartite Steiner $(1+\eps)$-spanner for $P$ with $|E|$ edges, there is also a $(2+2\eps)$-spanner for $P$ with at most $|E|$ edges.
\end{theorem}
\begin{proof}
    For each Steiner point $s$, find the closest point $s'$ that is connected to it and connect each point that was connected to $s$ to $s'$ instead.
    For any two points $p,q$ whose distance in the Steiner spanner was given by $\dist(p,s) + \dist(s,q)$, the distance in the spanner is at most
    \[ \dist(p,s') + \dist(s',q) \leq \dist(p,s) + 2\dist(s, s') + \dist(s,q) \leq 2\dist(p,s) + 2\dist(s,q). \]
    Thus, all distances increase by a factor of at most two, which means this is a $(2+2\eps)$-spanner.
\end{proof}

\subsubsection{Lower bound}
Note that Theorem~\ref{thm:Hbipartite} proves Theorem~\ref{thm:main} for $d \geq 3$.
However, for $d=2$ we only have the degree bound $\bigO{\frac{1}{\eps} \log\frac{1}{\eps}}$, which is much larger than the $\bigO{\frac{1}{\sqrt{\eps}} \log\frac{1}{\eps}}$ bound in the Euclidean case.
Despite that, we can use the following lemma to show that Theorem~\ref{thm:Hbipartite} is in fact optimal up to factors $\log\frac{1}{\eps}$:
combined with \Cref{thm:2spanner} it rules out a bipartite Steiner $(1+\eps)$-spanner with $o(n / \eps)$ edges.
Incidentally, it also rules out a $2$-spanner with $o(n \log n)$ edges.
\begin{lemma}\label{lem:bipartitelowerbound}
    For any $n \geq 2$, let $\eps \in [0, \frac{1}{2 \ln n}]$ and $P \subset \Hyp^2$ be a set of $n$ points equally spaced around a circle of radius~$3 \ln n$.
    Then, any $(2+\eps)$-spanner for $P$ must have $\bigOm{n \log n}$ edges.
\end{lemma}
\begin{proof}    
    Number the points of $P$ consecutively and modulo $n$, then define the \emph{index distance} between $p_i$ and $p_j$ to be $\min\{|i - j|, n - |i - j|\}$.
    First, let us determine the relation between this index distance and the hyperbolic distance.
    For that, let $c$ be the centre of the circle and consider the points $p_0$ and $p_k$.
    Also, place a point $m$ halfway along segment $p_0p_k$, so that we get a right-angled triangle $cmp_0$.
    Hyperbolic trigonometry now gives $\sinh|p_0m| = \sin(\angle p_0cm) \cdot \sinh|p_0c|$.
    Therefore,
    \[ |p_0p_k| = 2|p_0m| = 2\arsinh\left( \sin\frac{\pi k}{n} \sinh(3 \ln n) \right). \]
    
    We will bound this by simpler functions from both directions.
    First of all, $2x/\pi \leq \sin x \leq x$ for $0 \leq x \leq \pi/2$ and $(1-e^{-2}) e^x / 2 \leq \sinh x \leq e^x / 2$ for $x \geq 1$, giving
    \begin{align*}
        2\arsinh\left( kn^2 (1-e^{-2}) \right) &\leq
        |p_0p_k| \leq
        2\arsinh\left( \frac{\pi kn^2}{2} \right). \\
    \intertext{Additionally, $\ln 2x \leq \arsinh x$ for $x \geq 0$ and $\arsinh x \leq \ln x + \arsinh\frac{\pi}{2} - \ln\frac{\pi}{2}$ for $x \geq \frac{\pi}{2}$, so}
        2\ln\left( 2kn^2 (1-e^{-2}) \right) &\leq
        |p_0p_k| \leq
        2\ln\left(\frac{\pi kn^2}{2} \right) + 2\arsinh\frac{\pi}{2} - 2\ln\frac{\pi}{2}, \\
    \intertext{which we can finally simplify and approximate as}
        1 + 2\ln k + 4 \ln n &<
        |p_0p_k| <
        3 + 2\ln k + 4 \ln n.
    \end{align*}

    From this it follows that all distances between pairs in $P$ are greater than $1 + 4 \ln n$, while the maximum distance in $P$ is given by the circle diameter $6 \ln n$.
    Now any $(2+\eps)$-spanner must have hop-diameter at most two, seeing as $(2+\eps) 6 \ln n \leq 12 \ln n + 3 = 3 (1 + 4 \ln n)$, so a path of $3$ hops (which is strictly longer than $3 (1 + 4 \ln n)$) would not be able to give a $(2+\eps)$-approximate path for any pair of points.

    Next, we will constrain this even further.
    Let $k \in \{1, \dots, \lfloor n/4 \rfloor\}$ and assume that the path in the $(2+\eps)$-spanner between $p_{-k}$ and $p_k$ goes through some $p_{\alpha k}$.
    We will prove a condition on $\alpha$.
    First,
    \begin{align*}
        |p_{-k}p_{\alpha k}| + |p_{\alpha k}p_k| &\leq (2+\eps)|p_{-k}p_k| \\
        2 + 2\ln(k|\alpha+1|) + 2\ln(k|\alpha-1|) + 8 \ln n &< (2+\eps)\left( 3 + 2 \ln 2k + 4 \ln n \right).
    \end{align*}
    Many terms cancel out, simplifying this as follows.
    We then use the bound $k \leq n/4$, followed by $\eps \leq \frac{1}{2\ln n}$ and finally $n \geq 2$.
    \begin{align*}
        2\ln|\alpha^2-1| &< 4\ln 2 + 4 + \left( 3 + 2 \ln 2k + 4 \ln n \right) \eps \\
        &\leq 4\ln 2 + 4 + \left( 3 - 2 \ln 2 + 6 \ln n \right) \eps \\
        &\leq 4\ln 2 + 7 + \left( 3 - 2 \ln 2 \right) / (2 \ln n) \\
        &\leq 4\ln 2 + 6 + \frac{3}{2 \ln 2}
    \end{align*}
    This can now be rewritten as
    \begin{align*}
        |\alpha^2-1| &< 4 e^{3 + 3 / (4 \ln 2)} \\
        |\alpha| &< \sqrt{1 + 4 e^{3 + 3 / (4 \ln 2)}} \\
        &< 16.
    \end{align*}

    Consider the points $P_{6k} = p_1, \dots p_{6k}$.
    We will prove that the spanner needs at least $2k$ edges with an endpoint in $P_{6k}$ and index length in $[k, 51k)$.
    To that end, first split $P_{6k}$ into parts $A$, $B$ and $C$ that contain $2k$ consecutive points each.
    Because the spanner must have hop-diameter at most two, we can say without loss of generality that the path between arbitrary $a \in A$ and $c \in C$ will go through an edge $av$ of index length at least $k$.
    Let $b \in B$ be the point with index halfway between $a$ and $c$; its index distance to $a$ is at most $3k$.
    Now by the inequality above, the index distance between $b$ and $v$ must be less than $48k$ and thus $av$ has index length less than $51k$.
    There are $2k$ disjoint pairs $a \in A, c \in C$ and each can only share its edge of index length at least $k$ with one other disjoint pair, so there must be $k$ of these edges.
    This proves the claim.

    Now, for convenience assume $n/6$ is a power of $51$ and let $k = 51^i$ with $i$ going from $0$ to $\log_{51}(n/6)$.
    For each $k$, we can split $P$ into $\frac{n}{6k}$ disjoint sets similar to $P_{6k}$, giving a total of $n/6$ edges that have index length in $[51^i, 51^{i+1})$.
    Thus, any $(2+\eps)$-spanner requires at least $(n/3)\log_{51}(n/6)$ unique edges.
\end{proof}

This leaves the question of whether a $2$-spanner with $\bigOd{n \log n}$ edges is possible, and in fact, it is.
Let us first introduce some concepts that will also play a key role in \Cref{sec:strongspanner}.
Given points $P \subset \Hyp^d$, we define the tree $T^\ell(P)$.
Its vertices are the cells $\{ \cell(p,T^\ell) \mid p \in P \} \cup \{ \lca(p,q,T^\ell) \mid p,q \in P\}$, as well as the points $P$.
There is an edge from any cell to the first cell above it and from any point to the cell containing it.
Note that this strongly resembles a $(d-1)$-dimensional Euclidean compressed quadtree; this lets us efficiently construct and maintain $T^\ell(P)$.

\begin{figure}
    \centering
    \includegraphics[scale=.8]{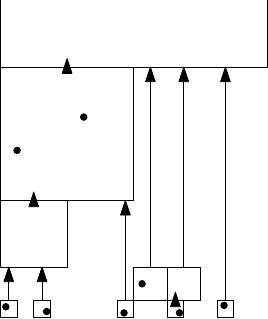}
    \caption{The tree $T^0(P)$ for the black points $P$.}
    \label{fig:tilingtree}
\end{figure}

\begin{lemma}\label{lem:dynamic_T(P)}
    We can maintain $T^\ell(P)$ such that updating it after a point insertion/deletion in $P$ takes $\bigO{d \log n}$ time.
\end{lemma}
\begin{proof}
    We store a linear quadtree \cite{Gargantini82, GeomAppAlg} for the point set $\pi_x(P) \subset \Reals^{d-1}$, which stores exactly the cells that are the lowest common ancestor in the $(d-1)$-dimensional Euclidean quadtree for some pair $p, q \in \pi_x(P)$.
    We also add the cells $\{ \cell(p,T^\ell) \mid p \in P \}$ to this structure.
    This gives a supergraph of $T^\ell(P)$ which still includes only $\bigO{n}$ cells.
\end{proof}

Besides $T^\ell(P)$, we will also need \Cref{lem:lambert}, which essentially says that a segment $pq$ should stick close to the corresponding path in $T^\ell$.
For small $\ell$ the result is perhaps not as surprising, but for larger $\ell$ the distance shrinks exponentially, which is a much stronger result and will be the basis of the near-optimal Steiner spanner in \Cref{sec:strongspanner}.

\begin{lemma}\label{lem:lambert}
    Let $C$ be a cell of $T^{1 + \log\log(1/\eps)}$, let $p$ be a point that lies in $C$ or one of its descendant cells, and let $q$ be a point that lies in an ancestor cell of the parent $C_1$ of $C$.
    Then, $pq$ passes within distance $\eps$ of $s_C = \left(\frac{\x(C) + \xm(C)}{2},\, \sqrt{\zd(C_1) \zu(C_1)}\right)$.
\end{lemma}
\begin{proof}
    Cells $C'$ of $T^{1 + \log\log(1/\eps)}$ have $w(C') = \frac{1}{2 \eps^2 \sqrt{d-1}}$ and $h(C') = 2\log\frac{1}{\eps}$, meaning $\zd(C') / \zu(C') = \eps^2$.
    Without loss of generality, assume that $x(p) = (0, \dots, 0)$ and $\zu(C) = 1$.
    Then, $z(p) \leq 1$ and $z(q) \geq 1/\eps^2$.
    Also, $q$ must lie in the same cell $C'$ of $T^{1 + \log\log(1/\eps)}$ as $((0, \dots, 0),\, z(q))$ and thus $|x_j(q)| \leq \frac{z(q)}{2 \eps^2 \sqrt{d-1}}$ for all horizontal axes $j \in \{1, \dots, d-1\}$.
    In the halfspace model, $pq$ is an arc from a Euclidean circle through $p$ and $q$ with its centre on $z=0$.
    Let $a$ be the point on $pq$ with $z(a) = z(s_C) = 1/\eps$.
    We will first find an upper bound on $\|x(a)\|$ to later also bound $|as_C|$.
    
    Note that decreasing $z(p)$ and increasing $\|x(q)\|$ both increase $\|x(a)\|$.
    Thus, we get an upper bound on $\|x(a)\|$ by setting $p$ to the origin and 
    \[q = ((\frac{z(q)}{2 \eps^2 \sqrt{d-1}}, \dots, \frac{z(q)}{2 \eps^2 \sqrt{d-1}}),\, z(q)).\]
    Let $S$ be the Euclidean circle that $pq$ is now an arc of and $r$ its radius.
    Then, the centre of $S$ is $c = ((\frac{r}{\sqrt{d-1}}, \dots, \frac{r}{\sqrt{d-1}}), 0)$.
    Now, $\|c - p\| = \|c - q\|$ and thus
    \begin{align*}
    r^2 &= (d-1)\left(\frac{r}{\sqrt{d-1}} - \frac{z(q)}{2 \eps^2 \sqrt{d-1}}\right)^2 + z(q)^2\\
    \left(\frac{r}{z(q)}\right)^2 &= \left(\frac{r}{z(q)} - \frac{1}{2\eps^2}\right)^2 + 1\\
    \frac{r}{z(q)} &= \eps^2 + \frac{1}{4\eps^2}.
    \end{align*}
    Note that $z(q) \geq 1/\eps^2$ implies $r \geq 1 + \frac{1}{4\eps^4}$ or equivalently $1/\eps^2 \leq 2\sqrt{r-1}$.

    Since $z(a) = 1/\eps$ is also the vertical Euclidean distance from $a$ to $c$, their horizontal Euclidean distance $\|x(c) - x(a)\| = \sqrt{r^2 - 1/\eps^2}$.
    Seeing as $x(a)$ and $x(c)$ lie in the same direction from the origin, we have 
    \begin{align*}
    \|x(a)\|
    &= \|x(c)\| - \|x(c) - x(a)\| \\
    &= r - \sqrt{r^2 - 1/\eps^2} \\
    &\leq r - \sqrt{r^2 - 2\sqrt{r-1}}.
    \end{align*}
    Using this inequality, we will prove that $\|x(a)\| < \frac{3}{4}$.
    This would hold if
    \begin{align*}
        r - \frac{3}{4} &< \sqrt{r^2 - 2\sqrt{r-1}} \\
        r^2 - \frac{3}{2}r + \frac{9}{16} &< r^2 - 2\sqrt{r-1} \\
        0  &< \frac{3}{2}(r-1) - 2\sqrt{r-1} + \frac{15}{16}.
    \end{align*}
    We can treat this as a quadratic function of $\sqrt{r-1}$.
    In particular, the coefficient of $(r-1)$ is positive and the discriminant is $-\frac{13}{8} < 0$, so the function is indeed always positive.
    This now proves that $\|x(a)\| < \frac{3}{4}$.
    
    Finally, $\|x(s_C)\| \leq w(C)\zd(C)\sqrt{d-1} / 2 = \frac{1}{4}$ and thus $|as_C| \leq \frac{\|x(a) - x(s_C)\|}{z(a)} < (\frac{3}{4} + \frac{1}{4}) \eps = \eps$, which proves the lemma.
\end{proof}

Now we have all we need for the $2$-spanner.
For small distances, it simply uses \Cref{thm:doublingspanner}, while for larger distances it uses a bipartite Steiner spanner with constant additive error and then applies the procedure of \Cref{thm:2spanner}.
The triangle inequality by itself is not enough to prove that this gives a $2$-spanner, but a more in-depth look at the hyperbolic distances is.

\twospanner*
\begin{proof}
    First, we use \Cref{thm:doublingspanner} to handle distances of up to $8$, the same way Theorem~\ref{thm:Hbipartite} uses Theorem~\ref{thm:doublingmain} for small distances.
    Namely, for each shift we partition the points $P$ based on level $4 + \log\log d$ of the shifted quadtree, then apply \Cref{thm:doublingspanner} with $\eps = \frac{1}{2}$.
    To handle distances larger than $2\arcosh 2$, we will first need the claim below.
    
    \begin{claim*}
        We can construct a bipartite Steiner spanner for $P$ with $d^{\bigO{d}} \cdot n \log n$ edges in $d^{\bigO{d}} \cdot n \log n$ time.
        For every pair $p,q \in P$ that are not in the same cell of $T^{4 + \log\log d}_i$ under any shift $i$, it connects both to a common Steiner point within distance $\frac{1}{4}$ of segment $pq$.
    \end{claim*}
    \begin{claimproof}
        We do the construction for every shift $i$ given by \Cref{lem:lca} and there we first consider the tree $T^{4 + \log\log d}_i(P)$.
        The construction is recursive.
        Given a tree $T$, we first find its \emph{centroid}, which is the node $v$ such that all connected components of $T \setminus\{v\}$ contain at most $n/2$ nodes.
        If $T$ has $n$ nodes, its centroid can be found greedily in $\bigO{n}$ time.
        Now, add a set of Steiner points for the cell $v$ and connect them to every point $p \in P$ that is still in $T$.
        In particular, we place these Steiner points as a $\frac{1}{4}$-covering on the boundaries of cell $v$ and the vertex-adjacent cells with the same $z$-coordinates.
        We then recurse on the connected components of $T \setminus\{v\}$.

        Since the number of nodes halves in each level of the recursion, this goes to depth at most $\log n$.
        In each level of the recursion, we spend a total of $\bigO{n}$ time on finding centroids and connect each point of~$P$ to $d^{\bigO{d}}$ Steiner points, since those Steiner points cover at most $3^{d-1}$ small-boundary cells each.
        Thus, we get the promised running time and edge count of $d^{\bigO{d}} \cdot n \log n$.
        
        Finally, given a pair $p,q \in P$ that are not in the same cell of $T^{4 + \log\log d}_i$ under any shift $i$, we need to prove that both are connected to a Steiner point $s$ within distance $\frac{1}{4}$ of segment $pq$.
        Note that by \Cref{lem:lca} there is a shift $i$ such that $pq$ goes through $C_\text{top} = \lca(p,q,T^{4 + \log\log d}_i(P))$.
        Additionally, the construction ensures that $p$ and $q$ are connected to Steiner points for some cell $C$ that separates them in $\lca(p,q,T^{4 + \log\log d}_i(P))$.
        If $C = C_\text{top}$, $p \in C$ or $q \in C$ then we are done.
        Otherwise, we know that $pq$ goes through some descendant and some ancestor of $C$, so we can use \Cref{lem:lambert} with $\eps = 1/d^{8}$, such that $1 + \log\log\frac{1}{\eps} = 4 + \log\log d$.
        Now, $pq$ must pass within distance $\eps$ of $C$, so in particular it must go through $C$ or one of the vertex-adjacent cells with the same $z$-coordinates.
        Thus, $pq$ always goes through a cell whose boundary has a $\frac{1}{4}$-covering of Steiner points that both $p$ and $q$ are connected to, which proves the claim.
    \end{claimproof}
    
    \begin{figure}
        \centering
        \includegraphics{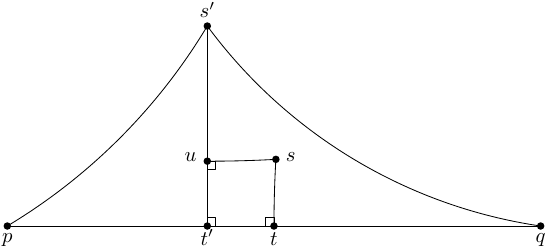}
        \caption{The $2$-spanner uses the path $p,s',q$ instead of the Steiner path $p,s,q$. Since $s$ is within distance $1/4$ to the line $pq$ and $|s's|\leq \min\{|sp|,|sq|\}$, we can show that $|ps'|+|s'q|\leq 2|pq|$.}
        \label{fig:2spanner}
    \end{figure}
    We now turn this Steiner spanner into a spanner using the procedure of \Cref{thm:2spanner}.
    To prove that this gives a $2$-spanner, take two points $p,q \in P$ whose shortest path in the Steiner spanner went through Steiner point $s$.
    In the spanner, $s$ is replaced by point $s' \in P$ where $|ss'| \leq \min\{|sp|, |sq|\}$.
    We will now prove that $|ps'| + |qs'| \leq 2|pq|$.
    To simplify the setting, we first use a (hyperbolic) rotation around $pq$ to put $s$ in the plane of $pqs'$; this might decrease the distance $|ss'|$ but leaves all other distances untouched.
    Let $t$ be the closest point to $s$ on $pq$ and $t'$ the closest point to $s'$ on the line through $p$ and $q$, then finally let $u$ be the closest point to $s$ on $s't'$.
    From the claim, $|ts| \leq \frac{1}{4}$.
    See \Cref{fig:2spanner} for an illustration.
    
    First assume that $|pt'| < 2$; we will prove that $|t's'| < 4$.
    Consider the quadrilateral $t'tsu$, which has right angles at $t'$, $t$ and $u$, making it a Lambert quadrilateral.
    This gives the formula $\sinh|ts| = \sinh|ut'| \cosh|us|$, which implies $|ut'| \leq |ts| \leq \frac{1}{4}$.
    From the triangle inequality we get $|us| \geq |tt'| - |t'u| - |st| \geq |tt'| - \frac{1}{2}$, as well as $|tt'| \geq |pt| - |pt'| > |pt| - 2$ and $|pt| \geq |ps| - |ts| \geq |ps| - \frac{1}{4}$.
    Putting these three together shows that $|us| > |ps| - \frac{11}{4}$.
    If $|ps| \leq \frac{11}{4}$ (making this bound meaningless), then already $|t's'| \leq |t's| + |ss'| \leq \frac{1}{4} + \frac{11}{4} < 4$ so we are done.
    Otherwise, consider the right-angled triangle $uss'$.
    The hyperbolic analogue to the Pythagorean theorem gives
    \begin{align*}
        \cosh|us'| &= \cosh|ss'|\ / \cosh|us'| \\
         &< \cosh|ps|\ / \cosh\left(|ps| - \frac{11}{4} \right).
    \intertext{Applying the identity $\cosh(x+y) = \cosh x \cosh y + \sinh x \sinh y$ now yields}
        \cosh|us'| &<
        \cosh|ps|\ / \left(\cosh|ps| \cdot \cosh\frac{11}{4} - \sinh|ps| \cdot \sinh\frac{11}{4} \right) \\
        &\leq 1 / \left(\cosh\frac{11}{4} - \sinh\frac{11}{4} \right) \\
        &= e^{11/4}.
    \end{align*}
    From here, $|us| < 3.5$, which gives $|t's'| = |t'u| + |us'| < 4$.
    Thus, $|ps'| + |qs'| \leq |pq| + 2|t's'| < |pq| + 8$.
    Since $|pq| \geq 8$, this proves that $|ps'| + |qs'| \leq 2|pq|$.
    Note that the same proof works when $|qt'| < 2$.

    If $|pt'| \geq 2$ and $|qt'| \geq 2$, then we use that $|ss'| \leq \frac{1}{4} + \frac{1}{2}|pq|$ to now bound $|t's'| \leq |t's| + |ss'| \leq \frac{1}{2} + \frac{1}{2}|pq|$.
    This lets us calculate the total length with again the hyperbolic Pythagorean theorem:
    \begin{align*}
        |ps'| + |qs'| &=
        \arcosh\left( \cosh|pt'| \cdot \cosh|t's'| \right) +
        \arcosh\left( \cosh|qt'| \cdot \cosh|t's'| \right) \\
        &\leq
        \arcosh\left( \cosh|pt'| \cdot \cosh\left(\frac{|pq|+1}{2} \right)\right) +
        \arcosh\left( \cosh|qt'| \cdot \cosh\left(\frac{|pq|+1}{2} \right)\right).
    \end{align*}
    Using the fact that $\arcosh x < \ln 2x$ for $x \geq 1$, we can upper bound this by
    \[
        \ln\left( 4 \cosh|pt'| \cdot \cosh|qt'| \cdot \cosh^2\left(\frac{|pq|+1}{2} \right)\right).
    \]
    Since each $\cosh$ has argument at least $2$, we can now use the bound $\cosh x \leq \cosh2 \cdot e^{x-2}$ for $x \geq 2$, which further upper bounds this as
    \begin{align*}
        \ln\left( 4 \cosh^4 1 \cdot e^{2|pq| - 7} \right).
    \end{align*}
    Note here that $4 \cosh^4 1 \cdot e^{-7} < 1$ and therefore $|ps'| + |qs'| < \ln(e^{2|pq|}) = 2|pq|$.
    Thus, the path from $p$ to $q$ through $s'$ is always at most twice as long as $pq$, which proves that this is a $2$-spanner.
\end{proof}

\subsection{Near-optimal Steiner spanner}\label{sec:strongspanner}
As shown by Lemma~\ref{lem:bipartitelowerbound}, the only way to improve our Steiner spanner further is if it is no longer bipartite, i.e.\ we must connect Steiner points to Steiner points.
However, with our current constructions these Steiner points occur in large groups and connecting all pairs among two such groups is not viable.
To deal with this gracefully, we will use a different method of placing Steiner points based on \Cref{lem:lambert}.

\paragraph{Construction.}
We first combine lemmas \ref{lem:lca} and \ref{lem:lambert} to get the next Steiner spanner $G_3(P,\eps)$, which has a very simple construction as the union of $\bigO{d}$ trees but only works for distances $\bigOm{\frac{1}{\eps} \log\frac{1}{\eps}}$.
Each of these trees comes from a shift $i$ given by Lemma~\ref{lem:lca} and is closely related to $T^\text{base}_i := T_i^{1 + \log\log(1/\eps)}(P)$:
for every cell $C$ in $T^\text{base}_i$, we place a Steiner point $s_C$
as defined in Lemma~\ref{lem:lambert},
then we add edges between these Steiner points as in $T^\text{base}_i$.
We also connect every point in $P$ to the Steiner point corresponding to the cell of $T^\text{base}_i$ it lies in.

\begin{lemma}\label{lem:verylargedist}
    Let $\eps \in (0, \frac{1}{2}]$ and $P$ be a set of points in $\Hyp^d$.
    There is a set of $\bigO{d}$ trees with $P$ as the leaves whose union is a Steiner spanner for $P$ with multiplicative distortion $1 + \eps$ and additive distortion $\bigO{\log\frac{1}{\eps}}$.
    In particular, if all points in $P$ have distance at least $\frac{1}{\eps}\log\frac{1}{\eps}$ to each other, this is a Steiner $(1 + \bigO{\eps})$-spanner.
\end{lemma}
\begin{proof}
    Take points $p,q \in P$ and let $\dist_G$ denote the distance in $G_3(P,\eps)$; we will prove that $\dist_G(p,q) \leq (1+\bigO\eps)|pq| + \bigO{\log\frac{1}{\eps}}$.
    By Lemma~\ref{lem:lca}, there is a shift $i$ where $pq$ goes through $C = \lca(p,q,T^\text{base}_i)$, thus $|ps_C| + |s_Cq| \leq |pq| + 4\diam(C)$, where $\diam(C) = \bigTh{\log\frac{1}{\eps}}$.
    The path in $G_3(P,\eps)$ from $p$ to $s_C$ will be some (possibly empty) sequence of $k$ Steiner points $s_{C_j}$.
    By Lemma~\ref{lem:lambert}, $ps_C$ must pass within distance $\eps$ of each Steiner point $s_{C_j}$, giving total error $k\eps$.
    Simultaneously, $|ps_C| \geq 2k \log\frac{1}{\eps}$ because it has to traverse the full height of $k$ cells.
    Thus, $\dist_G(p, s_C) \leq (1+\eps) |ps_C|$ and analogously $\dist_G(s_C, q) \leq (1+\eps) |s_Cq|$.
    Finally, $\dist_G(p,q) \leq (1+\eps) |pq| + \bigO{\log\frac{1}{\eps}}$.
    If $|pq| \geq \frac{1}{\eps} \log\frac{1}{\eps}$, then $\bigO{\log\frac{1}{\eps}} \leq \bigO{\eps}|pq|$ so this reduces to $(1+\bigO{\eps}) |pq|$.
\end{proof}

For smaller distances, we need to reduce the additive error $\bigO{\log\frac{1}{\eps}}$.
In particular, since distances up to $16 \log d$ can be handled by Theorem~\ref{thm:doublingmain}, we need to reduce it to $\bigO{\eps \log d}$.
To that end, we will incorporate ideas from \Cref{sec:bipartitespanner}.
We introduce $T_i^\text{fine}$, whose vertices are the cells of $T_i^{4 + \log\log d}$ contained in some cell of $T_i^\text{base}$ and whose edges again go from each cell to the first cell above it.
For each cell $C'$ of $T_i^\text{fine}$, we define $\cov(C')$ as a $\sqrt{\eps\diam(C')}$-covering of $C'|_{\{\frac{1}{6}\}}$.
Now, for every point $p \in P$ or $p = s_C$ for some cell $C$ in $T_i^\text{base}(P)$, consider its $2\log\frac{1}{\eps}$ ancestors $\cC$ in $T_i^\text{fine}$ and connect $p$ to $\cov(C')$ for every $C' \in \cC$.
This gives the following result, which proves the hyperbolic claim of \Cref{thm:main}:

\begin{theorem}\label{thm:Hmain}
    Let $\eps \in (0,\frac12]$ and $P$ be a set of $n$ points in $\Hyp^d$.
    There is a Steiner $(1+\eps)$-spanner for~$P$ with $d^{\bigO{d}} \eps^{(1-d)/2} \log\frac{1}{\eps} \cdot n$ edges.
    We can maintain it dynamically, such that each point insertion or deletion takes $d^{\bigO{d}}  \eps^{(1-d)/2} \log \frac{1}{\eps} \cdot \log n$ time.
\end{theorem}
\begin{proof}
    The edge count follows by construction: for each of $\bigO{d}$ shifts $i$, $T_i^\text{base}$ has $\bigO{n}$ vertices.
    For each of these, we consider at most $2\log\frac{1}{\eps}$ cells $C$ of $T_i^\text{fine}$ and connect to $\cov(C)$, where $C$ is a small-boundary cell so $|\cov(C)| = d^{\bigO{d}} \cdot (\eps \diam(C))^{(d-1)/2} = d^{\bigO{d}} \cdot \eps^{(d-1)/2}$.
    For dynamic updates, we maintain $T_i^\text{base}$ for each shift using \Cref{lem:dynamic_T(P)}, as well as an adjacency list representation of the Steiner spanner.
    The Steiner points and edges follow from the construction in $d^{\bigO{d}} \eps^{(1-d)/2} \log \frac{1}{\eps}$ time, which we can then each insert into the adjacency lists in a total of $d^{\bigO{d}} \eps^{(1-d)/2} \log \frac{1}{\eps} \cdot \log n$ time.

    To prove this is a Steiner spanner, we follow the same idea as Lemma~\ref{lem:verylargedist} except now $C = \lca(p, q, T_i^\text{fine})$ and we connect to a point $s \in \cov(C)$.
    In particular, let $p, q \in P$, then by Lemma~\ref{lem:lca} there is a shift $i$ where $p$ goes through $C = \lca(p, q, T_i^{4 + \log\log d})$.
    This cell $C$ is by construction also represented in $T_i^\text{fine}$.
    There is a (possibly empty) sequence of $k$ Steiner points $s_{C_j}$ where we can say $s_{C_0} = p$, each $s_{C_j}$ is connected to $s_{C_{j+1}}$, and $s_{C_k}$ is connected to all Steiner points in $\cov(C)$, such that either $k=0$ or the distance between $s_{C_k}$ and $C$ is at least $2\log\frac{1}{\eps}$:
    if we have a sequence where the distance between $s_{C_k}$ and $C$ is smaller than $2\log\frac{1}{\eps}$, then by construction $s_{C_{k-1}}$ will also be connected to all points in $\conv(C)$ so we can shorten the sequence.
    
    As shown in the proof of Lemma~\ref{lem:verylargedist}, $pq$ will pass within distance $\eps$ of each Steiner point $s_{C_j}$ and $|pq| \geq k\log\frac{1}{\eps}$.
    Additionally, there is a point $s \in \cov(C)$ such that (by the second part of Lemma~\ref{lem:lca}) $pq$ passes within distance $\sqrt{\eps \diam(C)}$ of $s$.
    Thus, by Lemma~\ref{lem:Hsplit2}, $|ps| + |sq| \leq |pq| + \eps\diam(C) \leq (1 + \eps)|pq|$.
    As shown, the path through the points $s_{C_j}$ is at most $1 + \eps$ times longer than $|ps|$, so in total $\dist_G(p,q) \leq (1 + 2\eps)|pq|$.
    In other words, this is a Steiner $(1+\eps')$-spanner for $\eps' = 2\eps$.
\end{proof}

\subsection{\texorpdfstring{$\eps$}{ε}-Additive Steiner spanner}\label{sec:additivespanner}
Using the same ideas, we can also make an additive instead of a multiplicative spanner.
This extends the $\bigO{\log d}$-additive Steiner spanner of \cite{additivespanner} and we follow their idea of using a \emph{transitive closure spanner}.
Given a directed graph, a $k$-transitive closure spanner adds edges that shortcut any existing path between two vertices to one of at most $k$ hops (see also the survey by Raskhodnikova~\cite{transclosure}).
Applying a transitive closure spanner to the trees of Lemma~\ref{lem:verylargedist} lets us get rid of the multiplicative distortion.
We make use of the following transitive closure spanner by Thorup~\cite{Thorup97}, where $\alpha(n)$ is the extremely slow-growing inverse Ackermann function.
It follows from his Algorithm L by taking $m = n$.

\begin{lemma}[Algorithm L of \cite{Thorup97}]\label{lem:transclosure}
    Let $T$ be a tree where its edges are directed towards the root.
    Then, there is an $\bigO{\alpha(n)}$-transitive closure spanner of $T$ that has $2n-1$ edges and can be found in $\bigO{n}$ time.
\end{lemma}

Putting this together gives \Cref{thm:additive}.
Note that \Cref{thm:main} already implies the same result without the factor $\alpha(n)$ for $\Sph^d$ and for points in a ball of radius $\bigOd{1}$ in $\Euc^d$ or $\Hyp^d$.

\additivespanner*
\begin{proof}
    We first construct $G_3(P,\eps / \alpha(n))$ (the trees of Lemma~\ref{lem:verylargedist}).
    On each of these trees, we apply the transitive closure spanner of \Cref{lem:transclosure}, so that for any pair $p,q \in P$ the additive distortion we accumulate along the path to their lowest common ancestor is $\bigO{\eps}$ (instead of being at most $\eps|pq|$).
    From here, we follow the construction of \Cref{thm:Hmain} again, which reduces the additive error of $\bigO{\log\frac{\alpha(n)}{\eps}}$ around the lowest common ancestor to $\bigO{\eps}$ as well.
    
    For small distances, we can make a Steiner $(1 + \eps / (16 \log d))$-spanner from Theorem~\ref{thm:doublingmain}.
    Because all distances are at most $16 \log d$, this adds at most $\eps$ to every distance.
\end{proof}

\section{Applications}\label{sec:applications}
\subsection{Quotient spaces}\label{sec:quotient}
By the Killing-Hopf theorem, any $d$-dimensional \emph{space form} (a constant-curvature Riemannian manifold that is complete and connected) is isometric to $\Sph^d$, $\Euc^d$ or $\Hyp^d$, modulo some isometry group $\Gamma$ (note however that not every isometry group gives a space form).
A classical example of this is the square flat torus, which is isometric to $\Euc^2$ modulo the translations from $\mathbb Z^2$.
Thus, results for $\Sph^d$, $\Euc^d$ and $\Hyp^d$ can be built upon to get results for these more general manifolds.

In particular, if $|\Gamma|$ is finite we can simply replace each point by its $|\Gamma|$ representatives in the covering space ($\Sph^d$, $\Euc^d$ or $\Hyp^d$), do the Steiner spanner construction there, then identify points in the covering space that are the same modulo $\Gamma$.
This is the case for, for example, $d$-dimensional elliptic space, which is isometric to $\Sph^d$ with antipodal pairs of points identified and thus has $|\Gamma| = 2$.

\begin{theorem}
    For any set $P$ of $n$ points on a manifold isometric to $\Sph^d / \Gamma$, $\Euc^d / \Gamma$ or $\Hyp^d / \Gamma$ for some finite isometry group $\Gamma$, there is a Steiner $(1+\eps)$-spanner with $d^{\bigO{d}} \cdot \eps^{(1-d)/2} \log\frac1\eps \cdot n |\Gamma|$ edges.
\end{theorem}

If $\Gamma$ is not finite then existing results still let us do something similar for orientable surfaces ($d=2$), assuming they still have constant curvature, are complete and are connected.
These surfaces fall into three categories based on their genus $g$: spheres ($g=0$), tori ($g=1$), and hyperbolic surfaces ($g \geq 2$).
In each case, $\bigO{g^{\bigO g}}$ representatives per point suffice.
For spheres $|\Gamma| = 1$ so this only requires $0^0 = 1$.
In the other cases, we use a \emph{Dirichlet domain}.
Let the surface be $S = \mathbb X^2 / \Gamma$ and take a base point $b \in S$.
Then, construct the Voronoi diagram for all representatives of $b$ in $\mathbb X^2$.
This diagram is highly regular: for any two cells there is an isometry in $\Gamma$ that maps one to the other.
Thus, modulo $\Gamma$ (i.e.\ on $S$) there is only one cell, which is the Dirichlet domain $\mathcal D_b$.
By construction, any shortest curve from $b$ to another point on $S$ does not cross the boundary of $\mathcal D_b$ (though it may end at the boundary).
From the Euler characteristic, it follows that $\mathcal D_b$ has at most $12 - 6g$ edges.
(If $\mathcal D_b$ has $k$ edges, then on the surface we have a graph with $e = k/2$ edges, $f=1$ face and $v \leq k/6$ vertices, since each vertex has degree at least $3$. The inequality now follows from $2 - 2g = v - e + f$.)

On a torus ($\mathbb X = \Euc$), Euclidean geometry implies that any Dirichlet domain $\mathcal D_{b'}$ is simply $\mathcal D_b$ translated by $b' - b$.
We will build a set $\widetilde P \subset \Euc^2$ of representatives of points in $P \subset S$.
For this, first fix one representative $\widetilde{\mathcal D_b} \subset \Euc^2$ of $\mathcal D_b$.
For each point $p \in P \subset S$, now add to $\widetilde P$ all representatives where the corresponding representative of $\mathcal D_p$ intersects $\widetilde{\mathcal D_b}$.
Now, take a point $q \in P$; it must have a representative $\tilde q \in \widetilde P \cap \widetilde{\mathcal D_b}$.
There will also be related representatives $\tilde p$ and $\widetilde{\mathcal D_p}$ where $\tilde q \in \widetilde{\mathcal D_p}$.
Since the shortest curve between $p$ and $q$ never crosses the boundary of $\mathcal D_p$, this means that $\tilde p \tilde q$ is a representative of this shortest curve.
Thus, this set $\widetilde P$ represents all shortest curves, which means a Steiner spanner for $\widetilde P$ becomes a Steiner spanner for $P$.
For a torus, the Dirichlet domains are either parallelograms or hexagons with parallel opposite sides; in both cases it can be shown that $|\widetilde P| \leq 4|P|$.
Similar observations have been used to construct the Delaunay triangulation on (three-dimensional) tori \cite{dolbilin1997periodic}.

Hyperbolic surfaces ($\mathbb X = \Hyp$) are much more challenging since translations no longer commute.
However, Despré, Kolbe and Teillaud \cite{DespreKT24} have proven that something similar still holds:

\begin{lemma}[Lemma 5 of \cite{DespreKT24}]
    Let $\mathcal D$ be a Dirichlet domain with $k$ edges.
    A shortest curve between two points crosses the boundary of $\mathcal D$ at most $k/2$ times.
\end{lemma}
This lemma lets us use $(k+1)^{k/2} = g^{\bigO g}$ copies for each point on $S$, leading to the following theorem:
\begin{theorem}
    Let $S$ be a constant curvature surface of genus $g$ that is closed, orientable and connected.
    For any set $P$ of $n$ points on $S$, there is a Steiner $(1+\eps)$-spanner with $\bigO{\frac{1}{\sqrt\eps} \log\frac1\eps \cdot n g^{\bigO g}}$ edges.
\end{theorem}
\begin{proof}
     As noted, for each genus we first find $\bigO{g^{\bigO{g}}}$ representatives of each point in $\Sph^2$, $\Euc^2$ or $\Hyp^2$.
     In the case $g=0$ we have the sphere so a single representative suffices, and when $g=1$ we have a torus with the set of representatives $\widetilde P$ described earlier.
     For genus $g \geq 2$ we first fix a Dirichlet domain, which must have $k \leq 12g - 6$ edges as noted earlier.
     Now consider all representatives reachable by crossing the domain boundary at most $k/2$ times.
     At every step there are $k+1$ distinct ways of (not) crossing the boundary, so this gives at most $(k+1)^{k/2} = g^{\bigO{g}}$ representatives per point.
     
     Next, use the construction of \Cref{thm:main} on the representatives, and finally identify all representatives of the same point again, taking the union of their edges.
     Since \Cref{thm:main} gives $\bigO{\frac{1}{\sqrt\eps} \log\frac1\eps \cdot \tilde n}$ edges where $\tilde n = \bigO{n g^{\bigO g}}$ here, this proves the theorem statement.
\end{proof}

Note that for hyperbolic surfaces we can also use Theorem~\ref{thm:additive} to get an additive approximation and Lemma~\ref{lem:verylargedist} to get a very sparse Steiner spanner for large distances.

\subsection{Approximate nearest neighbours}
As already noted in \cite{hyperquadtree}, if we can maintain a bipartite Steiner spanner dynamically then we only need a small modification to make it into a dynamic data structure that finds approximate nearest neighbours.
In particular, for each Steiner point $s$, we sort the points $p$ connected to it based on the distance $|sp|$.
To find a $(1+\eps)$-approximate nearest neighbour to a query point $q$, we then use \Cref{lem:lowdiam_implicit} or \Cref{lem:bipartite_implicit} to get the Steiner points it would be connected to and list the closest point for each of these Steiner points.
From this set, we take the point that is closest to $q$ and this has to be a $(1+\eps)$-approximate nearest neighbour.
\begin{theorem}\label{thm:anndyn}
    For $n$ points in $\Sph^d$, $\Euc^d$ or $\Hyp^d$ with $d \geq 3$, there is a data structure using $d^{\bigO{d}} \eps^{(1-d)/2} \log\frac1\eps \cdot n$ space, that can answer queries for a $(1+\eps)$-approximate nearest neighbour in $d^{\bigO{d}} \eps^{(1-d)/2} \log\frac1\eps + \bigO{d^2 \log n}$ time and perform updates (point insertions and removals) in $d^{\bigO{d}} \eps^{(1-d)/2} \log\frac1\eps \cdot \log n$ time.
\end{theorem}
Note that the case $d=2$ is a special case regardless, because there (in the static setting at least) we can get exact answers using Voronoi diagrams. Voronoi diagrams are efficiently computable in $\Euc^2$~\cite{BergCKO08}, $\Sph^2$~\cite{SphereVoronoi} and $\Hyp^2$~\cite{HyperVoronoi1, BogdanovDT14}, and provide $\Oh(\log n)$ query time via standard point location data structures~\cite{BergCKO08}.

As also noted in \cite{hyperquadtree}, we can use a very similar approach to maintain a dynamic approximate bichromatic closest pair.
\begin{theorem}
    For $n$ points in $\Sph^d$, $\Euc^d$ or $\Hyp^d$ with $d \geq 3$, there is a data structure using $d^{\bigO{d}} \eps^{(1-d)/2} \log\frac1\eps \cdot n$ space, that can maintain an $(1+\eps)$-approximate bichromatic closest pair while allowing updates (point insertions and removals) in $d^{\bigO{d}} \eps^{(1-d)/2} \log\frac1\eps \cdot \log n$ time.
\end{theorem}

\section{Conclusion}
\label{sec:conclusion}
We have given spanners of size $\Oh_{d}(\eps^{(1-d)/2}n)$ in $d$-dimensional spaces of constant curvature, most importantly, even in the non-doubling hyperbolic setting. 
Our results raise several natural questions. We believe that the following are especially promising open problems for future work:
\begin{itemize}
\item Let $\cR$ be a fixed $d$-dimensional Riemannian manifold (of non-constant curvature). Can one construct a Steiner spanner for any set of $n$ points in $\cR$ that has size $\widetilde \Oh_{\cR}(\eps^{(1-d)/2}n)$?
\item Can we get a polynomial dependence on the genus $g$ for Steiner spanners on (constant-curvature) surfaces?
\item Is there a near-linear size $\eps$-additive Steiner spanner in hyperbolic space that is dynamic? Can we remove the $\alpha(n)$ factor in its edge count?
\item Can we strengthen \Cref{lem:verylargedist}, for example by getting a smaller number of trees or already getting a Steiner $(1+\eps)$-spanner with $\bigO{dn}$ edges at distances $o(\frac{1}{\eps}\log\frac{1}{\eps})$?
\end{itemize}

\bibliographystyle{alphaurl}
\bibliography{bib}

\end{document}